\documentclass[sigplan, screen, pdfa]{acmart}

\copyrightyear{2024}
\acmYear{2024}
\acmConference[PPoPP '24]{The 29th ACM SIGPLAN Annual Symposium on Principles and Practice of Parallel Programming}{March 2--6, 2024}{Edinburgh, United Kingdom}
\acmBooktitle{The 29th ACM SIGPLAN Annual Symposium on Principles and Practice of Parallel Programming (PPoPP '24), March 2--6, 2024, Edinburgh, United Kingdom}\acmDOI{10.1145/3627535.3638492}
\acmISBN{}
\startPage{1}

\setcopyright{none}

\usepackage{microtype}
\usepackage[subtle]{savetrees}
\usepackage[caption=false]{subfig}
\usepackage{xspace}
\usepackage{amsmath}
\usepackage{amsthm}
\usepackage{amstext}
\usepackage{amsfonts}
\usepackage{tikz}
\usepackage{cleveref}
\usepackage{pgfplots}
\pgfplotsset{compat=1.9}
\usepgfplotslibrary{groupplots}
\usepackage{algorithm}
\usepackage{cases}
\usepackage{algpseudocode}
\usepackage{wrapfig}
\usepackage{url}
\usepackage{enumerate}
\usepackage{enumitem}
\usepackage[normalem]{ulem}
\usepackage{relate}
\usepackage{balance}
\usepackage{proc}
\usepackage{float}
\usepackage{flafter}
\usepackage{stfloats}
\usepackage{makecell}
\usepackage{listings}
\usepackage{placeins}
\usepackage[font=normal]{caption}
\pgfplotsset{
    discard if/.style 2 args={
        x filter/.append code={
            \edef\tempa{\thisrow{#1}}
            \edef\tempb{#2}
            \ifx\tempa\tempb
                
            \fi
        }
    },
    discard if not/.style 2 args={
        x filter/.append code={
            \edef\tempa{\thisrow{#1}}
            \edef\tempb{#2}
            \ifx\tempa\tempb
            \else
                
            \fi
        }
    },
}

\pgfplotsset{compat=1.8,
    /pgfplots/ybar legend/.style={
    /pgfplots/legend image code/.code={%
       \draw[##1,/tikz/.cd,yshift=-0.25em]
        (0cm,0cm) rectangle (3pt,0.8em);},
   },
}

\definecolor{mygreen}{rgb}{0,0.6,0}
\definecolor{mygray}{rgb}{0.5,0.5,0.5}
\definecolor{mymauve}{rgb}{0.58,0,0.82}

\definecolor{safe-black}{RGB}{0,0,0}
\definecolor{safe-olive}{RGB}{0,73,73}
\definecolor{safe-teal}{RGB}{0,146,146}
\definecolor{safe-pink}{RGB}{255,109,182}
\definecolor{safe-peach}{RGB}{255,160,110}
\definecolor{safe-plum}{RGB}{73,0,146}
\definecolor{safe-cerulean}{RGB}{0,109,219}
\definecolor{safe-lavender}{RGB}{182,109,255}
\definecolor{safe-sky}{RGB}{109,182,255}
\definecolor{safe-baby}{RGB}{182,219,255}
\definecolor{safe-brick}{RGB}{146,0,0}
\definecolor{safe-brown}{RGB}{146,73,0}
\definecolor{safe-orange}{RGB}{219,209,0}
\definecolor{safe-green}{RGB}{36,255,36}
\definecolor{safe-yellow}{RGB}{255,255,109}

\renewcommand{\epsilon}{\varepsilon}

\newcommand{\defn}[1]{\textit{#1}}
\newcommand{\system}{F-Graph\xspace}

\newcommand{\cpac}{C-PaC\xspace}
\newcommand{\othersystem}{\cpac}
\newcommand{\pactrees}{PaC-trees\xspace}
\newcommand{\ptree}{P-tree\xspace}
\newcommand{\ptrees}{P-trees\xspace}
\newtheorem{theorem}{Theorem}
\newtheorem{lemma}[theorem]{Lemma}

\definecolor{magenta4}{rgb}{0.5625,0,0.5625}
\definecolor{green4}{rgb}{0,0.5625,0}
\definecolor{orange4}{rgb}{0.98,0.31,0.09}

\newcommand{\todo}[1]{{\textcolor{red}{TODO: {#1}}}}

\newcommand{\figref}[1]         {Figure~\ref{fig:#1}}

\newcommand{\tabref}[1]        {Table~\ref{tab:#1}}
\newcommand{\secref}[1]         {Section~\ref{sec:#1}}
\newcommand{\lemref}[1]         {Lemma~\ref{lem:#1}}

\newcommand{\secreftwo}[2]      {Sections \ref{sec:#1} and~\ref{sec:#2}}
\newcommand{\secrefthree}[3]      {Sections \ref{sec:#1}, \ref{sec:#2}  and~\ref{sec:#3}}
\newcommand{\figreftwo}[2]      {Figures \ref{fig:#1} and~\ref{fig:#2}}

\newcommand{\rcircle}{\textsuperscript{\textregistered}\xspace}
\newcommand{\numcores}{\ensuremath{64}\xspace}
\newcommand{\numthreads}{\ensuremath{128}\xspace}

\newcommand{\rmat}{RMAT\xspace}

\newcommand{\minsmallrangespeedup}{$4\times$\xspace}
\newcommand{\rangespeedupcpac}{$1.2\times-10\times$\xspace}
\newcommand{\rangespeedupcpacavg}{$4\times$\xspace} 
\newcommand{\batchspeedup}{$3\times$\xspace}
\newcommand{\batchspeedupoverpam}{$1.5\times$\xspace} 
\newcommand{\rangespeedupoverpam}{$8.9\times-27.4\times$\xspace} 
\newcommand{\rangespeedupoverpamavg}{$20\times$\xspace} 

\newcommand{\graphalgspeedupcpac}{$1.2\times$\xspace}
\newcommand{\graphbatchspeedupcpac}{$2\times$\xspace}
\newcommand{\graphalgspeedupaspen}{$1.3\times$\xspace}
\newcommand{\graphbatchspeedupaspen}{$2\times$\xspace}
\newcommand{\spacesavingsoveraspen}{$0.6\times$\xspace}

\newcommand{\proc}{\texttt}

\newcommand{\cpp}{\texttt{C++}\xspace}
\newcommand{\searchsimple}{\proc{search(x)}\xspace}
\newcommand{\insertsimple}{\proc{insert(x)}\xspace}
\newcommand{\deletesimple}{\proc{delete(x)}\xspace}
\newcommand{\rangemap}{\proc{range\_map}\xspace}

\newcommand{\batchwork}{$O(k((\log^2 (n))/B + \log(n)))$}
\newcommand{\batchspan}{$O(\log(k) + \log^2(n))$}

\makeatletter
\algrenewcommand\ALG@beginalgorithmic{\footnotesize}
\makeatother

\renewcommand{\defn}[1]       {{\textit{\textbf{\boldmath #1}}}}
\renewcommand{\paragraph}[1]{\vspace{0.09in}\noindent{\bf \boldmath #1.}}

\algnewcommand{\algorithmicor}{\textbf{ or }}
\algnewcommand{\OR}{\algorithmicor}

\newcommand{\mymarginpar}[1]{}

\pgfplotsset{
    discard if/.style 2 args={
            x filter/.append code={
                    \edef\tempa{\thisrow{#1}}
                    \edef\tempb{#2}
                    \ifx\tempa\tempb
                        
                    \fi
                }
        },
    discard if not/.style 2 args={
            x filter/.append code={
                    \edef\tempa{\thisrow{#1}}
                    \edef\tempb{#2}
                    \ifx\tempa\tempb
                    \else
                        
                    \fi
                }
        },
}

\pgfplotsset{compat=1.11,
    /pgfplots/ybar legend/.style={
            /pgfplots/legend image code/.code={%
                    \draw[##1,/tikz/.cd,yshift=-0.25em]
                    (0cm,0cm) rectangle (3pt,0.8em);},
        },
}
\def\mystrut{\vphantom{hg}}
\pgfplotsset{
    legend style={font=\mystrut},
    y tick label style={font=\mystrut},
    x tick label style={font=\mystrut}
}

\begin{document}

\title{CPMA: An Efficient Batch-Parallel Compressed Set Without Pointers}         


\author{Brian Wheatman}
\affiliation{%
  \institution{Johns Hopkins University}
}
\email{wheatman@cs.jhu.edu}
\author{Randal Burns}
\affiliation{%
  \institution{Johns Hopkins University}
}
\email{randal@cs.jhu.edu}

\author{Ayd\i n Bulu\c{c}}
\affiliation{\institution{Lawrence Berkeley National Laboratory}}
\email{abuluc@lbl.gov}

\author{Helen Xu} \thanks{Work done while Helen Xu was at Lawrence Berkeley
  National Laboratory.}  \affiliation{%
  \institution{Georgia Institute of Technology}
}
\email{hxu615@gatech.edu}

\begin{abstract}
  This paper introduces the \defn{batch-parallel Compressed \\ Packed Memory Array}
  (CPMA), a compressed, dynamic, ordered set data structure based on the Packed
  Memory Array (PMA). 
  Traditionally,
  batch-parallel sets are built on pointer-based data structures such as trees
  because pointer-based structures enable fast parallel unions via pointer manipulation. When
  compared with cache-optimized trees, PMAs were slower to update but faster to
  scan.

  The batch-parallel CPMA overcomes this tradeoff between updates and scans by
  optimizing for cache-friendliness. On average, the CPMA achieves \batchspeedup faster batch-insert throughput and \rangespeedupcpacavg faster range-query throughput compared with 
  compressed \pactrees, a state-of-the-art batch-parallel set library based on
  cache-optimized trees.

  We further evaluate the CPMA compared with compressed \pactrees and Aspen, a state-of-the-art system, on a real-world
  application of dynamic-graph processing.  The CPMA is on average
  \graphalgspeedupcpac faster on a suite of graph
  algorithms and \graphbatchspeedupcpac faster on batch inserts when compared with compressed PaC-trees. Furthermore, the CPMA is on average \graphalgspeedupaspen faster on graph algorithms and \graphbatchspeedupaspen faster on batch inserts compared with Aspen.

\end{abstract}



\maketitle

\section{Introduction}\label{sec:intro}

The \defn{dynamic ordered set} data type (also called a key store) is one of the
most fundamental collection types and appears in many programming languages as
either a built-in basic type or in standard
libraries~\cite{BlelFeSu16,sun2018pam, DhulBlGu22}. Ordered sets enable
efficient scan-based operations (i.e., operations that use ordered iteration)
such as range queries and maps. This paper focuses on dynamic ordered sets which
also support updates (i.e., inserts and deletes).

Due to their role in large-scale data processing, dynamic ordered sets have been
targeted for efficient batch-parallel
implementations~\cite{frias2007parallelization,barbuzzi2010parallel, DhulBlGu22,
  DhulipalaBlSh19, sun2018pam, erb2014parallel, TsengDhBl19}.  Since point
operations (e.g., single-element insertion) are often not worth parallelizing
due to their sublinear complexity, modern libraries parallelize \defn{batch
  updates} that insert or delete multiple elements.  Direct support for batch
updates simplifies update parallelism and reduces the overall work of updates by
sharing work between updates.

Existing set\footnote{In this paper, we use ``sets'' to refer to dynamic
ordered batch-parallel sets.} implementations demonstrate the importance of optimizing for the
memory subsystem to achieve high performance.  Almost all fast
batch-parallel set implementations are built on pointer-based data structures
(e.g., trees)~\cite{DhulBlGu22, DhulipalaBlSh19, BlelFeSu16,
  frias2007parallelization,barbuzzi2010parallel, sun2018pam, erb2014parallel,
  TsengDhBl19}. Unfortunately, the main bottleneck in the scalability of these
sets is memory bandwidth limitations due to pointer chasing~\cite{BlelFeSu16,
  frias2007parallelization}.  Dhulipala \emph{et al.}~\cite{DhulipalaBlSh19,
  DhulBlGu22} mitigated these issues in trees by improving spatial locality via
blocking and compression.

Even with these improvements, cache-optimized trees inherently leave performance on the table because the
random memory accesses from following pointers are slower than contiguous memory
accesses~\cite{BenderCoFa19, PandeyWhXu21, WheatmanXu21}. In theory,
cache-friendly trees such as B-trees~\cite{BayerMc72} are asymptotically optimal
in the classical external-memory model~\cite{AggarwalVi88} for both updates and
scans. Empirically, array-based data structures support scans over
$2\times$ faster than tree-based data structures due to prefetching and the cost
of pointer chasing~\cite{PandeyWhXu21,WheatmanXu23}.

\paragraph{Exploiting sequential access with PMAs}
This paper introduces a work-efficient \defn{batch-parallel Compressed \\Packed
  Memory Array} (CPMA) based on the Packed Memory Array
(PMA)~\cite{itai1981sparse, BendDeFa00,bender2005cache}, a dynamic array-based
data structure optimized for cache-friendliness (i.e., spatial locality).
The PMA appears in domains such as graph
processing~\cite{ShaLiHe17,wheatman2018packed, DeLeoBo19, WheatmanXu21,
  de2021teseo,PandeyWhXu21, wheatman2021streaming}, particle
simulations~\cite{durand2012packed}, and computer
graphics~\cite{toss2018packed}.

Existing PMAs suffer from low update throughput compared to
batch-parallel trees because they lack direct algorithmic support for parallel batch
updates~\cite{WheatmanXu21}. At a high level, batch-update algorithms can be
implemented with unions/differences~\cite{BlelFeSu16}. Previous
work~\cite{de2019packed} introduced a serial batch-update algorithm for PMAs
based on local merges 
but stopped short of parallelization.

Supporting theoretically and practically
efficient parallel unions in a PMA requires novel algorithmic development
because existing parallel batch-update algorithms rely heavily on pointer
adjustments, which do not easily translate to contiguous memory layouts.

As an additional optimization, this paper adds
compression to PMAs. Previous work on compressed blocked trees~\cite{DhulBlGu22,
  DhulipalaBlSh19} demonstrates the potential for compression to alleviate
memory bandwidth limitations by reducing the number of bytes transferred. This
paper applies the same techniques to PMAs.

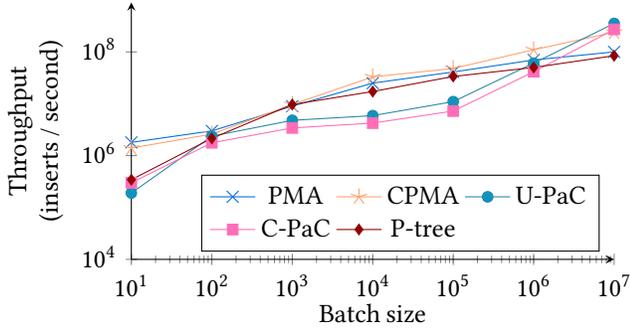
\begin{figure}
  \centering
  \begin{tikzpicture}
    \begin{axis}[
        width=8cm, height=5cm,
        axis lines = left,
        xlabel = Batch size,
        ylabel style={align=center},
        ylabel = {Throughput\\ (inserts / second)},
        xmode=log,
        ymode=log,
        ymin=10000,
        ymax=900000000,
        cycle list name=exotic,
        legend pos=south east,
        legend columns=3,
        xlabel shift={-6pt},
      ]

      \addplot[mark=x, safe-cerulean,mark options={scale=1.6}]
      coordinates{(10,1.8E+6)(100,3.0E+6)(1000,9.0E+6)(10000,2.5E+7)(100000,4.1E+7)(1000000,7.0E+7)(10000000,1.0E+8)};
      \addlegendentry{PMA}

      \addplot[mark=star, safe-peach,mark options={scale=1.6}]
      coordinates{(10,1.4E+6)(100,2.6E+6)(1000,9.7E+6)(10000,3.3E+7)(100000,4.8E+7)(1000000,1.1E+8)(10000000,2.4E+8)};
      \addlegendentry{CPMA}

      \addplot
      coordinates{(10,190187.585)(100,2370290.813)(1000,4798069.18)(10000,5902969.583)(100000,10996690.55)(1000000,60722279.37)(10000000,354021149.2)};
      \addlegendentry{U-PaC}

      \addplot[safe-pink, mark=square*]
      coordinates{(10,299300.3513)(100,1768338.341)(1000,3419497.606)(10000,4239253.64)(100000,7237729.822)(1000000,41433825.83)(10000000,271678593.4)};
      \addlegendentry{C-PaC}

     \addplot[safe-brick, mark=diamond*]
      coordinates{(10,343510.0947)(100,2153225.74)(1000,9688088.127)(10000,17223986.05)(100000,33616959.62)(1000000,50184578.88)(10000000,84501353.29)};
      \addlegendentry{\ptree}

    \end{axis}
  \end{tikzpicture}
  \vspace{-.8cm}
  \caption{Insert throughput as a function of batch size.}
  \label{fig:batch-micro}
\end{figure}

\begin{figure}
  \centering
  \begin{tikzpicture}
    \begin{axis}[
        width=8cm, height=5cm,
        axis lines = left,
        xlabel = Expected elements processed per search,
        ylabel style={align=center},
        ylabel = {Throughput \\(elements / second)},
        xmode=log,
        ymode=log,
        ymin=10000000,
        ymax=50000000000,
        cycle list name=exotic,
        legend pos=south east,
        legend columns=3,
        xlabel shift={-6pt},
      ]





      \addplot[mark=x, safe-cerulean,mark options={scale=1.6}]
      coordinates{(1.490116119,4.5E+8)(2.980232239,8.7E+8)(5.960464478,1.7E+9)(11.92092896,3.2E+9)(23.84185791,5.1E+9)(47.68371582,6.6E+9)(95.36743164,8.8E+9)(190.7348633,8.3E+9)(381.4697266,1.3E+10)(762.9394531,1.5E+10)(1525.878906,1.5E+10)(3051.757813,1.6E+10)(6103.515625,1.6E+10)(12207.03125,1.7E+10)(24414.0625,1.7E+10)(48828.125,1.8E+10)(97656.25,1.8E+10)(195312.5,1.8E+10)(390625,1.9E+10)(781250,1.9E+10)(1562500,1.9E+10)};
      \addlegendentry{PMA}

      \addplot[mark=star, safe-peach,mark options={scale=1.6}]
      coordinates{(1.490116119,2.0E+8)(2.980232239,4.1E+8)(5.960464478,8.1E+8)(11.92092896,1.6E+9)(23.84185791,2.9E+9)(47.68371582,5.1E+9)(95.36743164,8.2E+9)(190.7348633,1.2E+10)(381.4697266,1.5E+10)(762.9394531,1.8E+10)(1525.878906,2.0E+10)(3051.757813,2.2E+10)(6103.515625,2.3E+10)(12207.03125,2.3E+10)(24414.0625,2.4E+10)(48828.125,2.4E+10)(97656.25,2.4E+10)(195312.5,2.4E+10)(390625,2.4E+10)(781250,2.4E+10)(1562500,2.4E+10)};
      \addlegendentry{CPMA}

      \addplot
      coordinates{(1.490116119,5.0E+7)(2.980232239,1.0E+8)(5.960464478,2.0E+8)(11.92092896,3.7E+8)(23.84185791,6.4E+8)(47.68371582,9.5E+8)(95.36743164,1.2E+9)(190.7348633,1.2E+9)(381.4697266,2.1E+9)(762.9394531,4.1E+9)(1525.878906,7.0E+9)(3051.757813,1.0E+10)(6103.515625,1.3E+10)(12207.03125,1.6E+10)(24414.0625,1.6E+10)(48828.125,1.7E+10)(97656.25,1.8E+10)(195312.5,1.8E+10)(390625,1.9E+10)(781250,1.9E+10)(1562500,1.9E+10)};
      \addlegendentry{U-PaC}

      \addplot[safe-pink, mark=square*]
      coordinates{(1.490116119,4.7E+7)(2.980232239,1.0E+8)(5.960464478,1.8E+8)(11.92092896,3.5E+8)(23.84185791,5.8E+8)(47.68371582,8.5E+8)(95.36743164,1.0E+9)(190.7348633,1.2E+9)(381.4697266,1.5E+9)(762.9394531,2.3E+9)(1525.878906,4.1E+9)(3051.757813,7.0E+9)(6103.515625,1.0E+10)(12207.03125,1.3E+10)(24414.0625,1.6E+10)(48828.125,1.7E+10)(97656.25,1.8E+10)(195312.5,1.9E+10)(390625,1.9E+10)(781250,1.9E+10)(1562500,1.9E+10)};
      \addlegendentry{C-PaC}

      \addplot[safe-brick, mark=diamond*]
      coordinates{(1.490116119,5.9E+7)(2.980232239,1.1E+8)(5.960464478,1.9E+8)(11.92092896,2.9E+8)(23.84185791,4.0E+8)(47.68371582,4.9E+8)(95.36743164,5.5E+8)(190.7348633,5.8E+8)(381.4697266,6.0E+8)(762.9394531,6.1E+8)(1525.878906,6.2E+8)(3051.757813,6.2E+8)(6103.515625,6.3E+8)(12207.03125,6.4E+8)(24414.0625,6.5E+8)(48828.125,6.6E+8)(97656.25,6.7E+8)(195312.5,6.8E+8)(390625,6.9E+8)(781250,6.9E+8)(1562500,6.9E+8)};
      \addlegendentry{\ptree}


    \end{axis}
  \end{tikzpicture}
  \vspace{-.8cm}
  \caption{Range query throughput as a function of range size.}
  \label{fig:map-range-micro}
   \vspace{.4cm}
\end{figure}
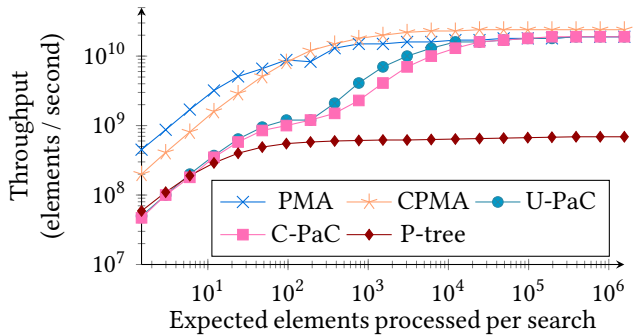

\paragraph{Results summary}
The CPMA's cache-friendliness translates into performance: the CPMA overcomes
the traditional tradeoff between updates and queries in trees and
PMAs.~\figreftwo{batch-micro}{map-range-micro} demonstrate that the CPMA
achieves on average \batchspeedup faster batch-insert throughput and
\rangespeedupcpacavg faster range-query throughput compared to Parallel Compressed trees (\pactrees)~\cite{DhulBlGu22}. PaC-trees\footnote{\pactrees are implemented in a
  library called CPAM, but we use ``\pactrees'' and ``\cpac'' in this paper to
  avoid confusion with ``CPMA.'' Similarly \ptrees are implemented in a library
  called PAM.} are a state-of-the-art batch-parallel set implementation based on
cache-optimized blocked trees. We also found that the uncompressed PMA achieves
on average \batchspeedupoverpam faster batch-insert throughput and
\rangespeedupoverpamavg faster range-query throughput when compared to
\ptrees~\cite{sun2018pam} (PAM), an efficient batch-parallel set implementation
based on binary trees. We compare the PMA with \ptrees because they are both
uncompressed.  Finally, CPMAs use similar space to compressed \pactrees but at least
$2\times$ less space than uncompressed PMAs. Furthermore, PMAs use about
$3\times$ less space than \ptrees.

To understand the improved locality of the PMA compared to \pactrees, we
measured\footnote{We added 100 million elements serially in batches of 1 million
  and measured the cache misses with \texttt{perf stat}.} the number of cache
misses during batch inserts in both. The PMA incurs at least $3\times$ fewer
cache misses when compared to \pactrees because the PMA takes advantage of
contiguous memory access as can be seen in \tabref{intro-cache-misses}.

Furthermore, to demonstrate the applicability of the CPMA, we built
\system\footnote{\system uses only one compressed PMA (a flat array) to store
  the graph. The F in \system comes from the musical key of F, which has one
  flat.}, a dynamic-graph-processing system built on the CPMA because PMAs have
been used extensively in graph processing~\cite{ShaLiHe17, wheatman2018packed,
  DeLeoBo19, WheatmanXu21, de2021teseo,PandeyWhXu21,
  wheatman2021streaming}. \system is on average \graphalgspeedupcpac faster on a
suite of graph algorithms and \graphbatchspeedupcpac faster on batch updates
compared to \cpac, a dynamic-graph-processing framework based on compressed
\pactrees. \system uses marginally less space to store the graphs when compared
to \cpac.  We also evaluate Aspen~\cite{DhulipalaBlSh19}, a state-of-the-art
dynamic-graph-processing framework based on compressed blocked trees. We find
that \system is on average \graphalgspeedupaspen faster on graph algorithms,
\graphbatchspeedupaspen faster on batch updates, and uses about
\spacesavingsoveraspen the space when compared to Aspen.

\begin{table}[t]
  \begin{center}
    \begin{tabular}{@{}ccccc@{}}
      \hline
      \textit{Workload}         & \textit{U-PaC}~\cite{DhulBlGu22} &
     \textit{C-PaC}~\cite{DhulBlGu22} & \textit{\textbf{PMA}}  &
      \textit{\textbf{CPMA}}
      \\
      \hline
 L1 misses & 3.1E9 & 2.2E9 & 8.9E8 & 7.0E8 \\
 L3 misses & 3.1E8 & 7.6E7 & 9.6E7 & 1.1E7 \\
      \hline
    \end{tabular}
    \caption{Cache misses incurred during batch inserts. 
      }\vspace{-.5cm}
    \label{tab:intro-cache-misses}
  \end{center}
\end{table}

\subsection*{Contributions}
\begin{itemize}
\item The design and analysis of a theoretically efficient parallel batch-update
  algorithm for PMAs (and for CPMAs).
\item An implementation of the PMA and CPMA with the parallel batch-update
  algorithm in \cpp.
\item An evaluation of the PMA/CPMA with \pactrees and
  P-trees. 
\item An evaluation of \system, a dynamic-graph-processing system based on the
  CPMA, compared to \cpac and Aspen.
\end{itemize}

\begin{table*}[t]
  \centering
 \resizebox{\textwidth}{!}{%
  \begin{tabular}{llllll}
    \hline
                        & \multicolumn{2}{c}{\textit{Work}}      & \phantom{a}              &
    \multicolumn{2}{c}{\textit{Span}}                                                                                                                                            \\
    \cmidrule{2-3} \cmidrule{5-6}
    \textit{Operation}  & \textit{CPMA (this work)}                 &
    \textit{Compressed PaC-tree }~\cite{DhulBlGu22}
                        &                                        &  \textit{CPMA (this work)}    & \textit{Compressed PaC-tree}~\cite{DhulBlGu22}                                     \\
    \hline
    Insert/delete       & $O((\log^2 (n))/B + \log (n))^\dagger$ & $O(\log(n) + P)$         &                                                & $O(\log(n))$
                        & $O(\log(n) + P/B)$                                                                                                                                     \\
    Batch insert/delete & \batchwork\textsuperscript{$\dagger$}  &
    $O(k\log(n/Pk))^\dagger$
                        &                                        & \batchspan               & $O(\log(n/P)\log k +
      P/B)$                                                                                                                                                                      \\
    Search              & $O(\log(n))$                           & $O(\log(n) + P)$         &                                                & $O(\log(n))$ & $O(\log(n) + P/B)$ \\
    Range query         & $O(\log(n) + r/B)$                     & $O(\log(n) + (P + r)/B)$ &                                                & $O(\log(n))$ &
    $O(\log
      (n) + P/B)$                                                                                                                                                                \\
    \hline
  \end{tabular}
  }
  \caption{Asymptotic bounds for operations in a CPMA and compressed
    PaC-tree. We use $B$ to denote the cache-line size, $k$ to denote the size
    of the batch, $r$ to denote the number of elements returned by the range
    query, and $P$ to denote the user-specified tree node block size in the
    PaC-tree (called $B$ in~\cite{DhulBlGu22}). Bounds with
    \textsuperscript{$\dagger$} are amortized. All bounds are $\Omega(1)$.}
  \label{tab:bounds-summary}
\end{table*}


\section{Related work}\label{sec:related}
This section describes how this work relates to prior work in parallel data structures. Specifically,
it discusses concurrent versus batch-parallel data structures and their use cases.

There is extensive work on concurrent data structures such as trees~\cite{aksenov2017concurrency, arbel2018getting, braginsky2012lock, bronson2010practical, MaoKoMo12, ellen2010non, kung1980concurrent, natarajan2014fast}, skip lists~\cite{herlihy2006provably, pugh1998concurrent}, and PMAs~\cite{WheatmanXu21}. Concurrent data structures are mostly orthogonal to this paper, which focuses on batch-parallel data
structures. Existing concurrent trees typically support mostly point operations (i.e., linearizable inserts/deletes and finds), whereas the CPMA in this paper also supports range queries and maps (and associated operations such as filter and reduce). Some recent work studies range queries in concurrent trees~\cite{BasinBoBr20, FatoPaRu19}. On the other hand, concurrent trees support asynchronous updates, which are more
general than batch updates because batch updates require a single writer. Therefore, fairly
comparing concurrent and batch-parallel data structures on update throughput is
challenging as their update functionalities are different.

The PAM paper~\cite{sun2018pam} demonstrated that batch-parallel binary trees
can achieve orders of magnitude higher insertion throughput compared to concurrent cache-optimized trees~\cite{MaoKoMo12, wang2018building, zhang2016reducing}. 

Batch-parallel and concurrent data structures are suited for different use
cases. For example, batch-parallel data structures have recently become popular for
both practical~\cite{EdgigerMcRi12, KyrolaBlGu12, MackoMaMa15,
  DhulipalaBlSh19} and theoretical~\cite{AcarAnBl19,NowiOn21,
  dhulipala2020parallel,DhulLiSh21,FerrLu94,TsengDhBl19} dynamic-graph
algorithms and containers. They are well-suited for applications with a large number of requests in a short time, such as stream processing or loop join~\cite{kim2010fast}. In contrast, concurrent data structures have been
used extensively in key-value stores for online transaction processing
applications that emphasize point operations such as put and get~\cite{ycsb}.


\section{Packed Memory Array}\label{sec:pma-prelim}
This section reviews the Packed Memory Array~\cite{BendDeFa00, itai1981sparse}
(PMA) data structure to understand the improvements in later sections. First, it
introduces the theoretical models used to analyze the PMA. It then describes the
PMA's structure and supported operations. Finally, it
details how to perform point updates in a PMA, which forms the basis for the
batch-update algorithm in~\secref{batch-updates}.

\paragraph{Analysis method}
~\tabref{bounds-summary} summarizes the bounds for key parallel operations in
the CPMA and compressed PaC-tree in the \defn{work-span
  model}~\cite[Chapter~27]{CLRS} and the \defn{external-memory
  model}~\cite{AggarwalVi88}.  The \defn{work} is the total time to execute the
entire algorithm in serial. The \defn{span} is the longest serial chain of
dependencies in the computation.
In the work-span model with binary forking, a parallel for loop with $k$
iterations with $O(1)$ work per iteration has $O(k)$ work and $O(\log(k))$ span.

The external-memory model introduces the cache-line-size parameter $B$ and measures algorithm cost in
terms of cache-line transfers. 

\paragraph{Design and operations}
The PMA maintains elements in sorted order in an array with (a constant factor
of) empty spaces between its elements.  Specifically, a PMA with $n$ elements
uses $N = \Theta(n)$ cells. The empty cells enable fast updates by reducing the
amount of data movement necessary to maintain the elements in order. The primary
feature of a PMA is that it stores data in contiguous memory, which enables fast
cache-efficient iteration through the elements.

A PMA exposes four operations:
\begin{itemize}
\item \insertsimple: inserts element \proc{x} into the PMA.
\item \deletesimple: deletes element \proc{x} from the PMA, if it exists.
\item \searchsimple: returns a pointer to the smallest element that is at least
  \proc{x} in the PMA.
\item \proc{range\_map(start, end, f)}: applies the function \texttt{f} to all
  elements in the range [\texttt{start}, \texttt{end}).
\end{itemize}

In this paper, we use the terms ``range map'' and ``range query''
interchangeably. Range queries can be implemented with the more general range
map, but we use the more popular term ``range query.''

The PMA supports
point queries (\texttt{search}) in $O(\log (n))$ cache-line transfers and
updates in $O((\log^2(n))/B + \log(n))$ (amortized and worst-case) cache-line
transfers~\cite{bender2017file, BendDeFa00,bender2005cache,
  willard1982maintaining, willard1986good, willard1992density}.  The PMA
supports efficient iteration of the elements in sorted order, enabling fast
scans and range queries.  Specifically, the PMA supports the \rangemap
operation on $k$ elements in $O(\log (n) + k/B)$ transfers. It implements
\rangemap with a search for the first element in the range, then a
scan until the end of the range.
PMAs are asymptotically worse than \pactrees for
all inserts/deletes and match them for search and range queries (\tabref{bounds-summary}).

The PMA defines an implicit binary tree with leaves of size 
$\Theta(\log(N))$ cells. That is, the implicit tree has $\Theta(N/\log(N))$
leaves and height $\Theta(\log(N/\log(N)))$.  Every node in the PMA tree has a
corresponding \defn{region} of cells. Each leaf \\
$i \in \{0, \ldots, N/\log(N) - 1\}$ has the region $[i\log(N), (i+1)\log(N))$,
and each internal node's region encompasses all of the regions of its
descendants. The \defn{density} of a region in the PMA is the fraction of
occupied cells in that region.

Each node of the PMA tree has an \defn{upper density bound} that determines the
allowed number of occupied cells in that node. If an insert causes a node's
density to exceed its upper density bound, the PMA enforces the density bound by
redistributing elements with that node's sibling, equalizing the densities
between them.  The density bound of a node depends on its
height. 

\paragraph{Updating a PMA}
\begin{figure}
  \centering
  \includegraphics[width=.8\columnwidth]{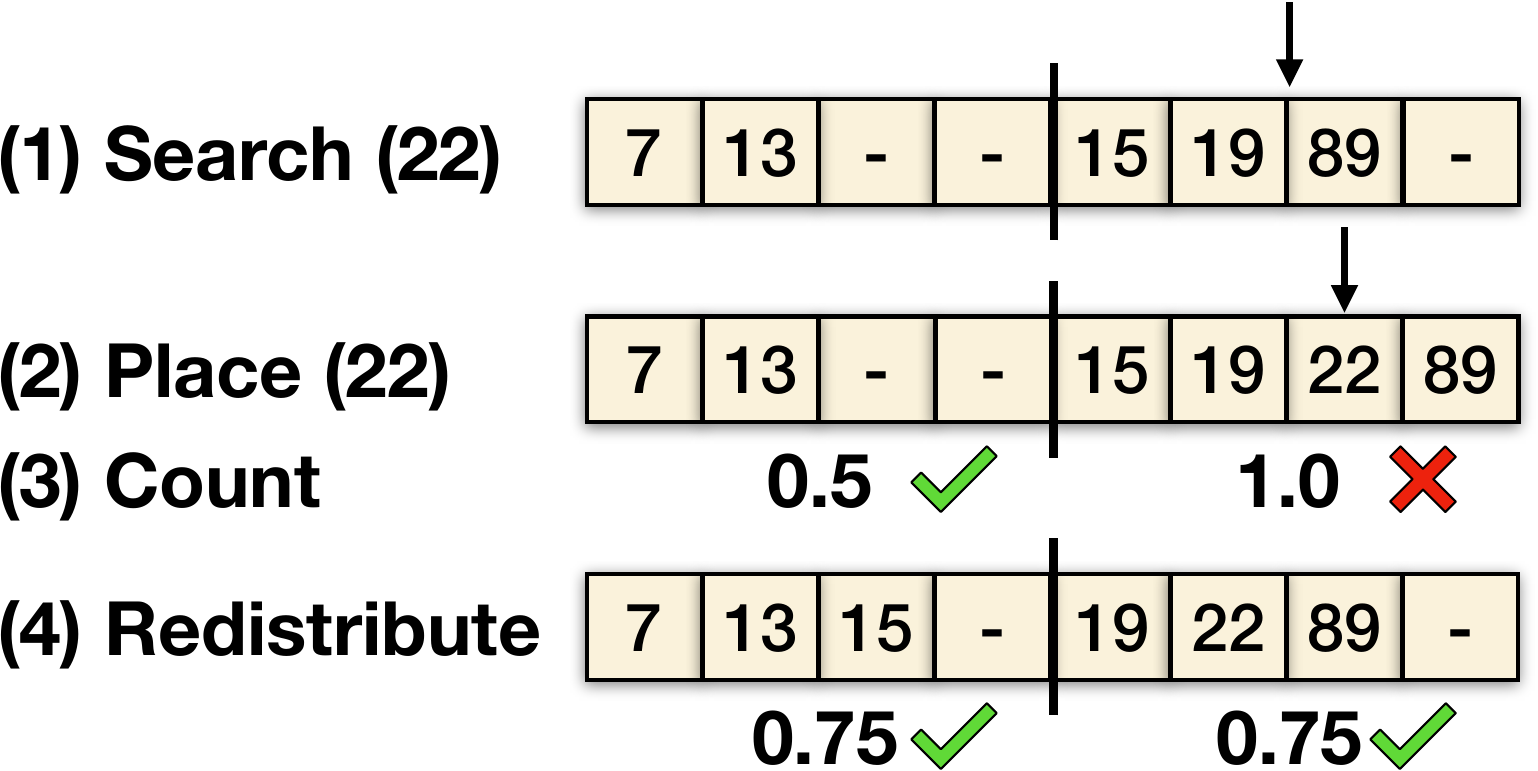}
  \caption{Example of an insertion in a PMA with leaf density bound of 0.9 and
    leaf size of 4.}
  \label{fig:pma-insert}
\end{figure}
A PMA maintains spaces between elements for efficient
updates. Since deletions are symmetric to insertions, we omit the discussion of deletes.

As shown in~\figref{pma-insert}, the four main steps in a PMA insertion are as
follows:
\begin{enumerate}
  \item \defn{Search} for the location that the element should be inserted into to
        maintain global sorted order.
  \item \defn{Place} the element at that location, potentially shifting some elements to
        make room.
  \item If the leaf that was inserted into violates its density bound, \defn{count} the density of nodes in the PMA to find a
        sibling to redistribute into.
  \item If necessary, \defn{redistribute} elements to maintain the correct distribution
        of empty spaces in the PMA.
\end{enumerate}

The four steps of a PMA insertion use the implicit tree to determine which leaf
to modify and which node to redistribute, if any. Steps (1) and (2) take $O(\log(n))$ cache-line transfers. Counting and redistributing elements (steps (3) and (4)) take $O((\log^2(n))/B)$
(amortized and worst-case) cache-line transfers~\cite{bender2017file,
  BendDeFa00,bender2005cache, willard1982maintaining, willard1986good,
  willard1992density}.

\paragraph{Resizing a PMA}
If the root-to-leaf traversal after an insert reaches the root and finds that
its density bound has been violated, the entire PMA is copied to a larger array
and the elements are distributed equally amongst the leaves of the new PMA.



\section{Parallel Batch Updates in a PMA}\label{sec:batch-updates}

Batching updates in a PMA improves throughput by sharing work between updates
and simplifying parallelization. This section describes how to apply batch
inserts in a PMA (batch deletes are symmetric).

We present a work-efficient parallel batch-insert algorithm for
PMAs. A \defn{work-efficient} parallel algorithm performs no more than a
constant factor of extra operations than the serial algorithm for the same
problem. Serially inserting $k$ elements into a PMA with $n$ elements takes
$O(k(\log(n) + (\log^2 (n))/B))$ cache-line transfers. Unfortunately, a naive
algorithm that parallelizes over the inserts is not work-efficient because it
may recount densities to determine which regions to redistribute. Supporting work-efficient batch inserts requires careful algorithm design
to avoid redundant work. Finally, we conclude the section with a microbenchmark
that demonstrates the serial and parallel scalability of batch inserts in PMAs.

The \defn{batch-insert problem} for PMAs takes as input a PMA with $n$
elements and a batch with $k$ sorted elements to insert. An unsorted batch can
be converted into a sorted batch in $O(k\log (k))$ work.

The optimal strategy for applying a batch of updates depends on the size of the
batch. At one extreme, if $k$ is small (e.g., $k < 100$), the overheads from the
batch-update algorithm outweigh the benefits, so point updates are more
efficient than batch updates. At the other extreme, if $k$ is large (e.g.,
$k \geq n/10$), the optimal algorithm is to rebuild the entire data structure
with a linear two-finger merge . The batch-insert algorithm for PMAs performs
local merges to address the intermediate case between these two extremes.

The parallel \defn{batch-insert algorithm} for PMAs applies a
batch of updates efficiently in the case where neither point insertions nor a
complete merge are the best options. Therefore, we focus on the
case\footnote{The assumption $k = o(n)$ is only used in the proofs to ensure
  sorting does not dominate the total cost.} where $\omega(1) = k = o(n)$.

The batch-insert algorithm consists of three phases: (1) a batch-merge
phase, (2) a counting phase, and (3) a redistribute phase. The phases proceed in
serial, but each phase is parallelized internally. The phases adapt the steps of
a PMA insertion described in~\secref{pma-prelim} to the batch setting. The
batch-merge phase combines the search and place steps, and the counting and
redistribute phases generalize their counterparts from point inserts.

\begin{figure}
  \centering
  \includegraphics[width=\columnwidth]{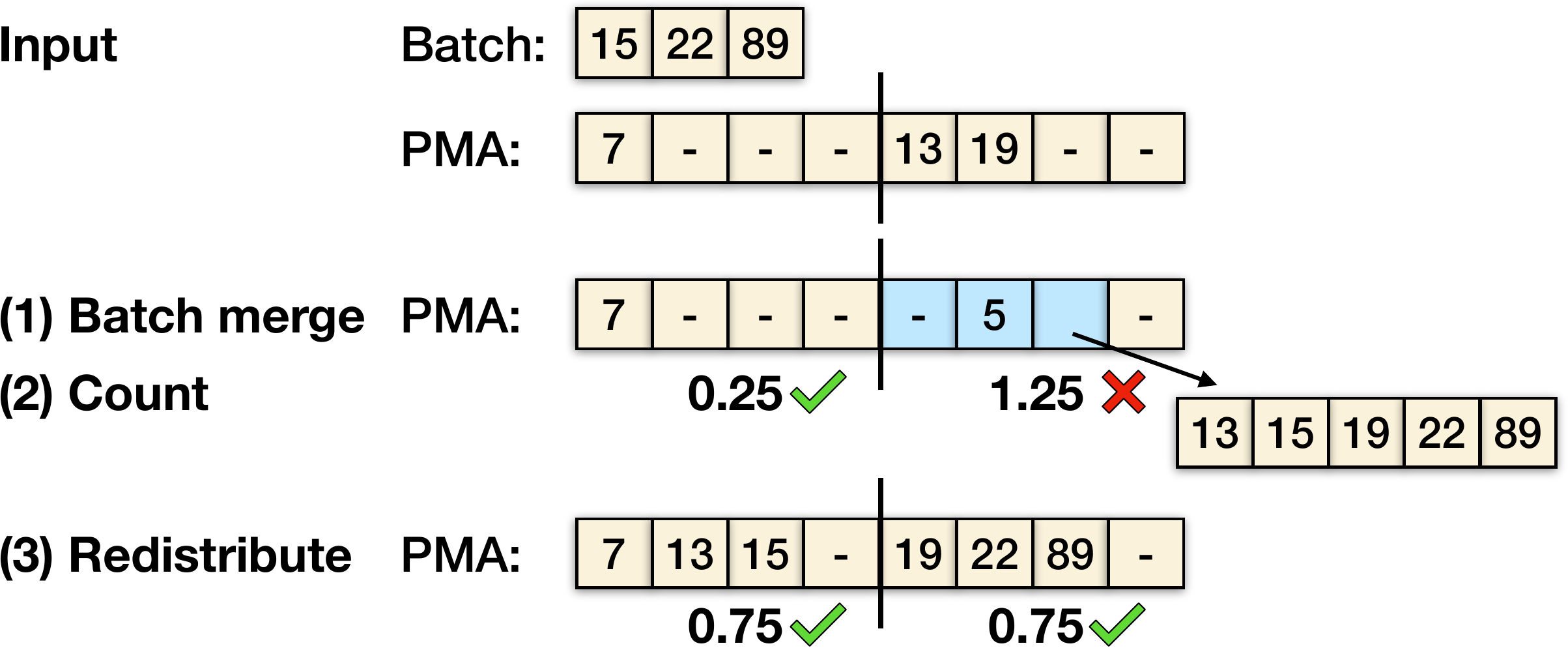}
  \caption{Example of batch insertion in a PMA with leaf density bound of 0.9
    and leaf size of 4.  After the merge, there are more elements in the second
    leaf than the leaf size, so the number of elements is stored in the leaf,
    and the elements are stored out-of-place until the redistribute.}
  \label{fig:batch-insert-fig}
\end{figure}

At a high level, the parallel \defn{batch-merge phase} divides the PMA and the
batch recursively into independent sections and operates independently on those
sections. Each recursive step first merges elements from the batch into one leaf
of the PMA, and then recurses down on the remaining left and right portions of
the batch. This recursive merge phase is inspired by recursive join-based
algorithms in batch-parallel trees~\cite{Adams92,Adams93, DhulBlGu22,
  DhulipalaBlSh19, BlelFeSu16, sun2018pam}. Existing join
algorithms for tree layouts rely on pointer adjustments which do not easily translate into array layouts.

At each step of the recursion, we perform a PMA search for the midpoint (median)
of the current batch and merge the relevant elements from the batch destined for
that leaf into the target leaf. Finding the bounds in the batch of all elements
in that leaf takes two searches (one backwards and one forwards).  Once the
endpoints have been found, we fork the merge of all relevant elements from the
current batch into the target leaf. If the number of elements destined for a
leaf is sufficiently large, we use a parallel merge algorithm with
load-balancing guarantees to achieve parallelism~\cite{AklSa87}.  Finally, we
recurse on the remaining left and right sides of the batch in parallel.

\begin{lemma}
  \label{lem:batch-insert-merge-work-span}
  Given a batch of $\;k$ sorted elements, the work of the batch-merge phase is
  $O(k\log(n))$, and the span is \\ $O(\log(k)\log(n))$.
\end{lemma}

\begin{proof}
  The height of the recursion is $O(\log(k))$, and each search in the PMA takes
  $O(\log(n))$ work. Finding the first and last element in the batch destined
  for the leaf takes $O(\log(k))$ work with exponential searches, which is
  smaller than \\ $O(\log(n))$. Therefore, the work and span of finding the
  bounds for the recursion is $O(\log(k)\log(n))$.

  In the worst case for the work, each element in the batch could be destined
  for a different leaf, so the total work of merging $k$ elements into $k$
  leaves is $O(k\log(n))$, which is asymptotically larger than the work to
  perform the recursion.

  The worst-case span for any one of the merges is $O(\log(k))$, so the total
  worst-case span of all the merges is $O(\log^2(k))$, which is less than
  $O(\log(k)\log(n))$.
\end{proof}

\begin{figure*}[t]
  \begin{center}
    \includegraphics[width=.9\textwidth]{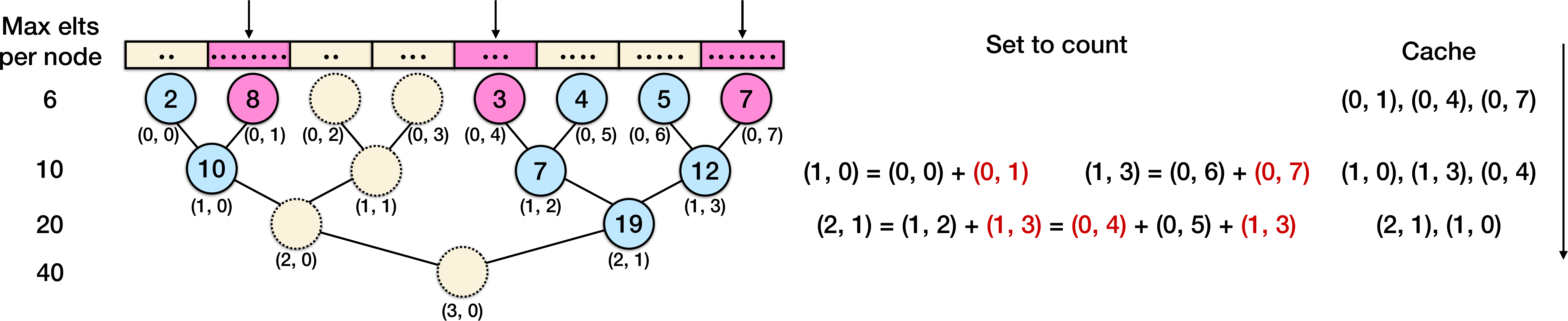}
  \end{center}
  \caption{An example of the work-efficient counting algorithm for batch
    updates. The blocks at the top represent the PMA leaves and the dots
    represent elements in the PMA. The pink blocks with arrows represent leaves
    that were touched during a batch update. The tree below the PMA is the
    implicit PMA tree of nodes labeled with a tuple of (height, index) (indices
    are assigned left to right). The blue solid circles represent PMA nodes that
    must be counted because their sibling or child violated its density. The tan
    dotted circles represent PMA nodes that did not need to be
    counted.}
  \label{fig:counting-alg}
\end{figure*}

When merging elements from the batch into a leaf, the target leaf may overflow
because it does not have enough space to hold all the elements destined for it.
To resolve this issue, the batch merge copies all elements into separate memory
and keeps the size of the extra memory as well as a pointer to it in the
leaf. This extra data is then cleaned up after the merge during the
redistribution phase. ~\figref{batch-insert-fig} illustrates a batch merge, leaf
overflow, and subsequent redistribution.

During the recursive batch merge, we keep track of all modified PMA leaves in a
thread-safe set for use in the counting and redistribution phases.

\paragraph{Counting phase}
After merging all elements into the PMA, the batch insert algorithm performs a
\defn{counting phase} where it finds the PMA nodes that violate their
density bounds for later redistribution.  The $O(\log^2(n)/B)$ work bound for point
insertions in the PMA comes from amortized analysis of the counting and
redistribution phase~\cite{itai1981sparse}, so efficiently counting and
redistributing in the batch-parallel setting is critical to achieving
work-efficiency.  To understand how to avoid redundant work, we start with a presentation of an efficient serial algorithm and describe how simply parallelizing this algorithm can lead to extra work.  We then present our work-efficient parallel algorithm.


An efficient serial batch algorithm must count each required cell exactly once. The
algorithm starts with the set of leaves that were touched in the batch-merge
phase. The ancestors of these leaves in the implicit PMA tree may need to be
redistributed.  The serial algorithm checks every leaf in turn. If a leaf
violates its density bound, the algorithm then walks up the implicit PMA tree
from that leaf until it finds a node that respects its density bound. Finding
the density of a node involves counting all of its descendants. By caching every
result and checking the cached results before counting, the serial algorithm
counts every required cell exactly once.

Unfortunately, simply parallelizing this serial algorithm over the leaves is not
work-efficient because the algorithm may recount PMA nodes whose densities have
not been cached yet. Therefore, the parallel algorithm may recount the same
region more than a constant number of times if many leaves share the same
ancestor to be redistributed.

To resolve this issue, we devise a new work-efficient parallel \defn{counting
  algorithm} that counts each required PMA cell exactly once.
~\figref{counting-alg} presents a worked example of this counting algorithm.
The counting algorithm takes as input the leaves that were modified in the batch
merge and outputs the set of PMA nodes that need to be redistributed.

This parallel algorithm avoids redundant work by processing the levels serially
from the leaves to the root and saving any counts for later lookups by nodes in
higher levels. At each level, we maintain a thread-safe set of nodes that need
to be counted. This set is initialized with the leaves that were affected by the
batch merge. The levels are processed serially, but all nodes at each level are
processed in parallel. If any node at some level $i$ exceeds its density bound,
the algorithm adds its parent to the set of nodes to be counted at level
$i+1$. The algorithm terminates when there are no more nodes to be counted, or
it has reached the PMA root.

\begin{lemma}
  \label{lem:batch-insert-counting-work}
  The parallel counting algorithm is work-efficient.
\end{lemma}

\begin{proof}
  The parallel counting algorithm caches results from each counted region as it
  processes the levels of the PMA tree. Due to the serial iteration of levels,
  all nodes to be counted at a level are counted in parallel. When a node $x$
  needs to be counted, no other node $y$ at that level will need to count any of
  $x$'s descendants since the set ensures that $x\neq y$.  All descendants of
  $x$ have either already been counted and cached, or will be counted exactly
  once and cached to avoid recounting.
\end{proof}

\begin{lemma}
  \label{lem:batch-insert-counting-span}
  The span of the parallel counting algorithm is $O(\log^2(n))$.
\end{lemma}

\begin{proof}
  The counting algorithm serially iterates over at most $O(\log(n))$ levels of
  the PMA because the height of the PMA tree is bounded by $O(\log(n))$. In the
  worst case, for each level $i$, the algorithm may have to recurse down $i$
  levels to count, so the worst-case span of traversing the PMA tree levels is:
  $\sum_{i = 0}^{\log(n)} i = O(\log^2(n))$. The PMA leaves are $O(\log(n))$
  cells each, so the total span of counting is $O(\log^2(n))$.
\end{proof}

\paragraph{Redistribution phase}
Once the counting phase has identified the correct regions to redistribute, the
PMA redistributes regions by performing two copies of the relevant data. The first copy
packs the regions to redistribute from the PMA into a buffer, and the second copy
equalizes the densities in the regions to redistribute by spreading the elements
evenly from the buffer into the target leaves.

\begin{lemma}
  \label{lem:batch-redistribute-work-span}
  Given a batch of $\;k$ sorted elements, the work of the redistribute phase is
  $O((k\log^2 (n))/B))$ amortized cache-line transfers, and the worst-case span
  is  $O(\log^2(n))$.
\end{lemma}

\begin{proof}
  The work of the redistribute phase is bounded above by the work of the
  counting phase, because the number of elements that need to be redistributed is
  at most the number of elements that need to be counted.
  From~\lemref{batch-insert-counting-work}, the counting step is work-efficient,
  so it takes no more than the serial amortized work bound of
  $O((k\log^2 (n))/B))$ cache-line transfers.

  The span of the redistribute phase is bounded above by $O(\log^2 (n))$ because
  there are at most $n$ independent sections to redistribute of size $n$
  each. Redistributing each one involves a parallel copy in and out, which has
  span $O(\log(n))$.
\end{proof}

\paragraph{Putting it all together}
Analyzing the entire batch-insert algorithm just involves summing the work and
span of the merge, counting, and redistribute phases of the batch-insert
algorithm.

\begin{theorem}
  \label{thm:batch-insert-work}
  The batch-insert algorithm for PMAs inserts a batch of $\;k$ sorted
  elements in $O(k(\log(n) + \log^2(n)/B))$ amortized work and $O(\log^2(n))$
  worst-case span.
\end{theorem}

\begin{table}[t]
  \begin{center}
    \resizebox{\columnwidth}{!}{%
    \begin{tabular}{@{}cccccc@{}}
      \hline
      \begin{tabular}{@{}c@{}}\textit{Batch} \\ \textit{size}\end{tabular} & \begin{tabular}{@{}c@{}}\textit{Serial} \\ \textit{TP}\end{tabular} &
      \begin{tabular}{@{}c@{}}\textit{Speedup over} \\ \textit{serial point}  \end{tabular}
                                 & \begin{tabular}{@{}c@{}}\textit{Parallel}\\\textit{TP} \end{tabular} & \begin{tabular}{@{}c@{}}\textit{Speedup over}\\\textit{serial batch}  \end{tabular} & \begin{tabular}{@{}c@{}}\textit{Overall}\\\textit{speedup}  \end{tabular}               \\
      \hline
1---10&2.2E6&1.0&1.8E6&0.8&0.8\\
1E2&1.9E6&0.9&3.0E6&1.6&1.4\\
1E3&2.0E6&0.9&9.0E6&4.6&4.1\\
1E4&2.0E6&0.9&2.5E7&12.5&11.6\\
1E5&2.3E6&1.0&4.1E7&17.8&18.6\\
1E6&2.9E6&1.3&7.0E7&23.8&32.0\\
1E7&5.5E6&2.5&1.0E8&18.6&47.1\\
      \hline
    \end{tabular}
    }
    \caption{Throughput (TP) of serial and parallel batch insertions in the
      PMA. We use point insertions for small batches when the batch update
      algorithm does not provide practical benefits. Overall speedup is the
      speedup over serial point inserts.}
    \label{tab:pma-batch-micro}
  \end{center}
\end{table}

\paragraph{Batch insert microbenchmark} ~\tabref{pma-batch-micro} reports the
throughput of batch inserts as a function of batch size (using the setup
described in~\secref{pma-eval}). The PMA under test starts with 100 million
elements and we add an additional 100 million
elements. 

On one core, the batch-insert algorithm is up to $3\times$ faster than point
inserts when the batch is large.  Batch inserts in a PMA save computation over
point insertions by reducing the number of searches, the length of each search,
and the number of redistributions.  The batch algorithm performs only one binary
search per updated leaf because the remaining elements in the batch destined for
that leaf are merged in directly. Additionally, the searches are smaller because
they often search only a subsection of the PMA. Finally, the counting algorithm
combines ancestor ranges to redistribute in the PMA, potentially skipping levels
of redistribution.

Furthermore, ~\tabref{pma-batch-micro} shows that batch inserts in a PMA achieve
parallel speedup of up to about $19\times$ on \numcores cores (\numthreads
threads) as the batch size grows. The main bottleneck in the parallel scalability
of batch updates is memory bandwidth.~\secref{cpma} mitigates these
issues by adding compression to the batch-parallel PMA to reduce data movement.


\section{Compressed Packed Memory Array}\label{sec:cpma}

This section introduces, analyzes, and empirically evaluates the
\defn{Compressed Packed Memory Array} (CPMA).  Adding compression does not affect the PMA's asymptotic bounds.
Empirically, the CPMA achieves better parallel scalability than the PMA because
the parallel operations are memory-bound, so the CPMA's smaller size makes
better use of memory bandwidth.

\paragraph{Data compression techniques}
The CPMA exploits the fact that elements are stored in sorted order in a PMA to
apply delta encoding~\cite{smith1997scientist} to the elements.
\defn{Delta encoding} stores differences (deltas) between sequential
elements rather than the full element.  Given a sorted array $A$ of $n$
elements, delta encoding results in a new array $A'$ such that $A'_0 = A_0$ and
for all $i = 1, 2, \ldots, n-1$, $A'_i = A_i - A_{i-1}$.

These deltas can then be stored in byte codes,
which store an integer as a series of bytes~\cite{WittMoBe99, BlanfordBlKa04}. Each byte uses one bit as a
\emph{continue bit}, which indicates if the following byte starts a new element
or is a continuation of the previous element.

We use delta encoding with byte codes in the CPMA  because they are fast to decode and achieve most of
the memory savings of shorter codes~\cite{ShunDhBl15, DhulipalaBlSh19,
  BlanfordBlKa04}.

\begin{figure}[t]
  \begin{center}
    \includegraphics[width=\columnwidth]{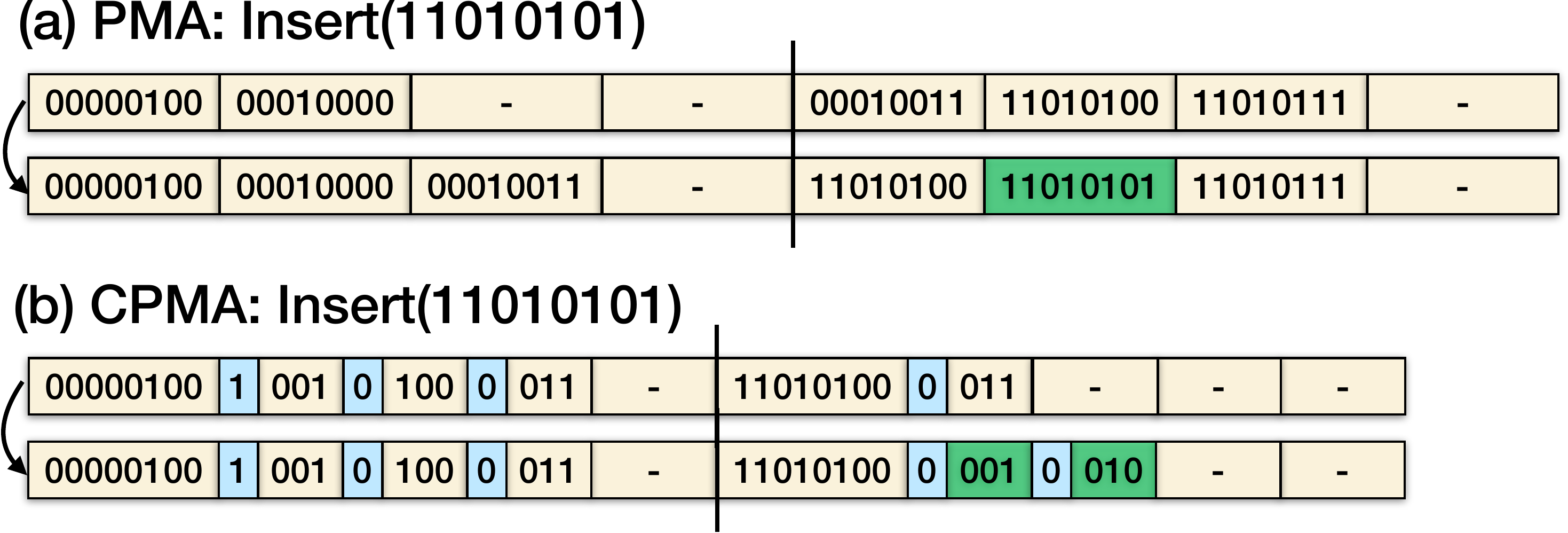}
  \end{center}
  \caption{An example of inserting the same element in a PMA and CPMA with the
    same elements. The density bound in all leaves is 0.9. Here,
    \texttt{sizeof(T)} is 8 bits, and a byte is 4 bits. The blue bits in the
    CPMA represent continue bits.  The green shaded cells in both the PMA and
    CPMA contain new data after the insert. The PMA redistributes its elements
    after the insertion, but the CPMA does not because the insertion did not
    violate the leaf density bound.}
    \vspace{.35cm}
  \label{fig:cpma-update}
\end{figure}

\paragraph{CPMA structure}
The CPMA maintains the same implicit binary tree structure as a
PMA and compresses the leaves.  Just like in the
PMA, a CPMA with $n$ elements and $N = \Theta(n)$ cells maintains leaves of size
$\Theta(\log(n))$ in order to achieve its asymptotic time bounds.  The CPMA
applies the \defn{packed-left} optimization, which packs elements to the left in
PMA leaves, for ease of compression~\cite{WheatmanXu21}. Packing the elements to
the left does not affect the PMA's (or CPMA's) asymptotic bounds\footnote{A
  traditional PMA redistributes all elements in a leaf after each
  insertion. Therefore, the packed-left property does not incur extra element moves over
  a regular PMA because both rearrange all elements in the leaf on each insert.}
because the bounds only depend on the density of the elements in the PMA
leaves~\cite{WheatmanXu21}.

A CPMA leaf stores its \defn{head}, or its first element, uncompressed, and
stores subsequent elements compressed with delta encoding and byte codes.  That
is, in a CPMA with elements of type \texttt{T}, the first \texttt{sizeof(T)}
bytes in each leaf contain the uncompressed head.  All following
cells take 1 byte each rather than \texttt{sizeof(T)} bytes.

The density bounds in a CPMA count byte density rather than element density. The
density in a CPMA node is the ratio of the number of filled bytes to the total
number of bytes available in the node.

\subsection*{CPMA Operations}

The CPMA maintains the same asymptotic bounds as the PMA
for point queries (searches) and point updates. Furthermore, compression does
not affect concurrency schemes for PMAs~\cite{WheatmanXu21} or the batch-update
algorithm from~\secref{batch-updates}.

The PMA's asymptotic bounds are derived from its implicit tree structure and
related density bounds. The main change in the CPMA is the compression of
each individual leaf, which does not affect the high-level implicit
tree structure.

The uncompressed head allows for efficient searching to find which leaf contains
an element.
The compressed leaves in the CPMA do not affect the high-level tree structure or
searches because each leaf can still be processed independently in $O(\log(n))$
work.

\paragraph{Point queries}
A CPMA on $n$ elements supports point queries in $O(\log (n))$ cache-line
transfers.  There are two steps in a point query in a CPMA: a binary search on
leaf heads, and then a pass through the leaf at the end of the binary search to
find the closest element. There are $O(n/\log(n))$ leaves, so a binary search
takes $O(\log(n))$ cache-line transfers. The leaf heads are stored uncompressed,
so there is no additional cost to perform the binary search on leaf heads
compared to a search in a PMA. After finding the target leaf, the CPMA performs
a search within that leaf. The size of each leaf is bounded by $O(\log(n))$, so
it takes $O(\log(n)/B)$ cache-line transfers to search a compressed leaf.

\paragraph{Point updates}
A CPMA on $n$ elements supports point updates in $O(\log(n) + (\log^2(n))/B)$
cache-line transfers. We will focus on the case of inserts, since deletes are
symmetric to inserts.~\figref{cpma-update} presents a worked example of the same
insert in a PMA and a CPMA.

The CPMA follows the same four steps of a PMA point update described
in~\secref{pma-prelim}.  We will focus on steps (2)-(4) (place, count, and
redistribute), since we already analyzed point queries.

 After performing a point query to find the target leaf, the CPMA places an
 element by adding a delta to the leaf and updating the following
 delta. Updating the leaf can be done in a single pass, which modifies up to
 $O(\log(n))$ cells because the size of the leaf is bounded by $O(\log(n))$. The
 CPMA matches the PMA's asymptotic bound on the number of cells modified during
 the place step.

Once the target leaf has been updated, the CPMA traverses up the leaf-to-root
path and redistributes any nodes that violate their density bounds just as in a
PMA. The amortized insert time bound comes from the checking and maintenance of
density bounds, which the CPMA supports in the same asymptotic cache-line transfers as a
PMA. Just as in a PMA, counting and redistributing in a CPMA takes cache-line transfers
linear in the size of the region.

\paragraph{Parallelizing the CPMA}
The compression in the CPMA does not conflict with existing lock-based
multiple-writer parallelism for PMAs~\cite{WheatmanXu21} because the locking
scheme depends on the implicit PMA tree structure and locking at a leaf
granularity. Furthermore, compression does not affect the theoretical
performance of concurrent PMAs because the CPMA also supports single-pass
operations within leaves. The CPMA supports multiple readers because reads are
non-modifying.

Finally, the batch-update algorithm in the CPMA is
identical to the batch-update algorithm for PMAs described
in~\secref{batch-updates}. The design and analysis of the batch-update algorithm
also depends only on single-pass operations on leaves.

\subsection*{Scalability analysis}
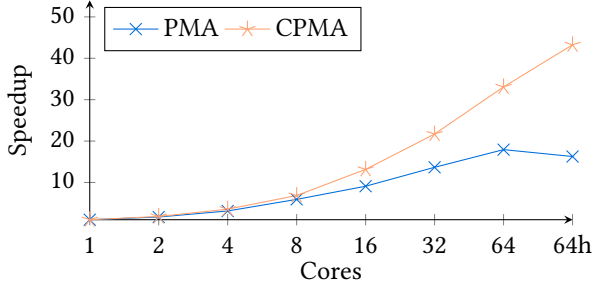
\begin{figure}
  \centering
  \begin{tikzpicture}
    \begin{axis}[
        width=8cm, height=4.5cm,
        axis lines = left,
        xlabel = Cores,
        ylabel style={align=center},
        ylabel = {Speedup},
        cycle list name=exotic,
        legend pos=north west,
        xmode=log,
        log basis x={2},
        legend columns=2,
        ymax=54,
        ytick={10,20,30,40,50},
        log ticks with fixed point,
        xlabel shift={-6pt},
        xtick={1,2,4,8,16,32,64,128},
        xticklabels = {1,2,4,8,16,32,64,64h},
      ]

      \addplot[mark=x, safe-cerulean,mark options={scale=1.6}]
      coordinates{(1,1)(2,1.65)(4,3.11)(8,5.91)(16,9.09)(32,13.66)(64,17.93)(128,16.27)};
      \addlegendentry{PMA}

      \addplot[mark=star, safe-peach,mark options={scale=1.6}]
      coordinates{(1,1)(2,1.78)(4,3.54)(8,6.84)(16,13.16)(32,21.70)(64,33.09)(128,43.25)};
      \addlegendentry{CPMA}


    \end{axis}
  \end{tikzpicture}
  \caption{Scalability of batch inserts in the PMA/CPMA. We use 64 to denote all physical cores and 64h to denote all 128 hyperthreads.}
  \label{fig:strong_scaling}
  \vspace{.25cm}
\end{figure}







We measure the scalability of both the PMA and CPMA on batch inserts and range
queries using the setup described in~\secref{pma-eval}. In each experiment, the
PMA and CPMA start with 100 million elements. In each batch-insert experiment,
we add 100 batches of 1 million elements each.  In each range-query experiment,
we perform 100,000 range queries in parallel where each query is expected to
return about $1.5$ million elements.  We measure the effect of core count on
performance of the PMA/CPMA. The extended version of the paper contains the raw
data.

~\figref{strong_scaling} shows that the CPMA achieves better scalability than the
PMA on batch inserts because compression maximizes the CPMA's usage of available
memory bandwidth. The PMA achieves up to $19\times$ speedup and the CPMA
achieves up to $43\times$ speedup for batch inserts on \numcores cores
(\numthreads threads). The CPMA achieves better batch insert throughput compared
to the PMA when the number of cores is sufficiently large (at least 16). When
the number of cores is too small, the additional computational overhead from
compression outweighs the benefits of decreased memory traffic.

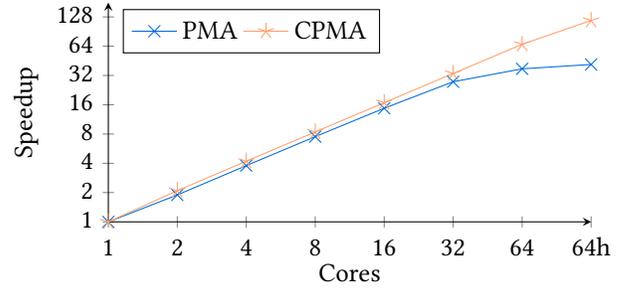
\begin{figure}
  \centering
  \begin{tikzpicture}
    \begin{axis}[
        width=8cm, height=4.5cm,
        axis lines = left,
        xlabel = Cores,
        ylabel style={align=center},
        ylabel = {Speedup},
        cycle list name=exotic,
        legend pos=north west,
        xmode=log,
        log basis x={2},
        ymode=log,
        log basis y={2},
        legend columns=2,
        ymax=180,
        ytick={1,2,4,8,16,32,64,128},
        log ticks with fixed point,
        xlabel shift={-6pt},
        xtick={1,2,4,8,16,32,64,128},
        xticklabels = {1,2,4,8,16,32,64,64h},
      ]

      \addplot[mark=x, safe-cerulean,mark options={scale=1.6}]
      coordinates{(1,1)(2,1.891538558)(4,3.784935965)(8,7.518593887)(16,14.72939785)(32,27.56712599)(64,37.43945548)(128,41.51270342)};
      \addlegendentry{PMA}

      \addplot[mark=star, safe-peach,mark options={scale=1.6}]
      coordinates{(1,1)(2,2.090915485)(4,4.201973781)(8,8.405794075)(16,16.79785831)(32,33.55162106)(64,66.91700144)(128,117.757703)};
      \addlegendentry{CPMA}


    \end{axis}
  \end{tikzpicture}
  \caption{Scalability of range queries in the PMA/CPMA. We use 64 to denote all physical cores and 64h to denote all 128 hyperthreads.}
  \label{fig:map_range_scaling}
   \vspace{.25cm}
\end{figure}

Similarly,~\figref{map_range_scaling}
demonstrates that the PMA achieves about $41\times$ speedup for range queries
and the CPMA achieves about $118\times$ speedup for range queries on \numcores
cores (\numthreads threads). The PMA's/CPMA's scalability on range queries is
much better than its scalability on updates because the queries proceed in
parallel and do not need to coordinate. The PMA's range query throughput in
terms of bytes transferred per second reaches the memory bandwidth on the
machine, but its overall range query throughput is limited because of the large
size per element.  The CPMA alleviates the memory bandwidth issue by decreasing
the size per element, enabling it to support more elements processed per byte
transferred.


\section{Evaluation}\label{sec:pma-eval}

To measure the improvements described in~\secreftwo{batch-updates}{cpma}, this
section evaluates the PMA/CPMA compared to uncompressed/compressed
\pactrees~\cite{DhulBlGu22} and \ptrees~\cite{sun2018pam} on range queries,
batch inserts, and space usage.  We use the terms ``U-PaC'' and ``C-PaC'' to
denote the uncompressed and compressed versions of \pactrees, respectively, in
this section.

This section then evaluates the CPMA, \cpac, and Aspen~\cite{DhulipalaBlSh19}, a
state-of-the-art dynamic-graph processing system based on compressed trees, on
an application benchmark of dynamic-graph processing because both PMAs and trees
appear frequently as dynamic-graph containers~\cite{ShaLiHe17,
  wheatman2018packed, DeLeoBo19, WheatmanXu21, de2021teseo,PandeyWhXu21,
  DhulipalaBlSh19, DhulBlGu22, wheatman2021streaming}. We introduce \system, a
system for processing dynamic graphs that uses the CPMA as its underlying data
structure.

Additional experiments and data tables can be found in an extended version of this paper~\cite{wheatman2023cpma}.

\paragraph{Microbenchmarks summary}
At a high level, the CPMA achieves the best of both worlds in terms of
performance. On average, it achieves \rangespeedupcpacavg faster range-query
throughput and \batchspeedup faster batch-insert throughput when compared to
compressed \pactrees. According to the theoretical prediction
in~\tabref{bounds-summary}, \pactrees asymptotically match or beat CPMAs for all
operations. However, in practice, the CPMA supports both fast queries and
updates due to its locality. Finally, CPMAs use about the same space as
compressed \pactrees, but they use less than half the space of uncompressed PMAs.
When compared with PAM, an uncompressed data structure, the uncompressed PMA
achieves \batchspeedupoverpam faster throughput for batch insertions and
\rangespeedupoverpamavg faster range query throughput.

\paragraph{Graph benchmark summary}
For graph workloads, we found that \system is on average \graphalgspeedupcpac
faster on a suite of graph algorithms, achieves \graphbatchspeedupcpac faster
throughput for batch updates, and uses marginally less space to store the graphs
compared to \cpac. Furthermore, \system is on average \graphalgspeedupaspen
faster on graph algorithms, achieves \graphbatchspeedupaspen faster throughput
for batch updates, and uses \spacesavingsoveraspen space to store the graphs
compared to Aspen.

\paragraph{Systems setup}
We implemented the PMA and CPMA as a \cpp library on top of the search-optimized
PMA~\cite{WheatmanXu23} and compiled them with \texttt{clang++-14}. To match the
parallelization method from the \pactrees library, we parallelized the PMA/CPMA
with the Parlaylib toolkit~\cite{blelloch2020parlaylib}.  The PMA and CPMA are
currently implemented as key stores (sets). The code can be found on 
\url{https://github.com/wheatman/Packed-Memory-Array.git}.

Each external library is compiled using the default configuration of
\texttt{g++-11} and Parlaylib (or PBBSlib, a precursor to Parlaylib, for Aspen)
for parallelization. They are each implemented in a \cpp library. We used the
in-place set mode of \ptrees and \pactrees for a fair comparison (although the
libraries also support a less efficient functional mode).  The \pactrees library block size is
set to the default for sets at 256, which corresponds to a maximum node size of
4108 bytes.
To initialize \pactrees, we used the library-provided recursive build
routine, which lays out the tree nodes non-contiguously in memory.

We also tested the Rewired PMA (RMA)~\cite{de2019packed} and compiled it with
the default provided scripts which used \texttt{clang++-14}. Since the RMA is
serial, there is no parallelization framework.

All experiments were run on a 64-core 2-way hyper-threaded Intel\rcircle
Xeon\rcircle Platinum 8375C CPU @ 2.90GHz with 256 GB of memory from
AWS~\cite{amazonaws}. Across all the cores, the machine has 3 MiB of L1 cache,
80 MiB of L2 cache, and 108 MiB of L3 cache. All performance results are the
average of 10 trials after a single warm up trial.

\subsection*{Evaluation on microbenchmarks}

We first evaluate the RMA, \ptrees, and \pactrees compared to the PMA/CPMA on a
suite of microbenchmarks.

\paragraph{Experimental setup}
We evaluate batch-update throughput first with $40$-bit uniform random numbers.  $40$-bit numbers gives a balance between the
compression ratio and the number of duplicates.  Uniform random is the
worst case for compressed data structures because it maximizes the deltas
between elements and therefore minimizes the compression ratio.  Uniform
random is also the worst case for batch inserts because it minimizes the
amount of shared work between updates that the algorithm can eliminate. However,
uniform random is the best case for 
redistributes in PMAs/CPMAs. 

We also evaluate batch-update throughput by
starting with $40$-bit uniform random numbers and then adding elements according
to a zipfian distribution. The zipfian distribution generates $34$-bit numbers
with skew parameter $\alpha = 0.99$ (parameter taken from the YCSB~\cite{ycsb}). For additional batch-insert experiments on skewed distributions, we test the
data structures on a skewed \rmat
distribution~\cite{ChakZhFa04} in the graph-processing application
benchmark at the end of this section.

We measure range-query performance of the data structure when it contains
100 million elements by performing 100,000 range queries in parallel. We
varied the size of the range queries across experiments.  We measure
batch-insert performance, by inserting 100
million elements in batches into a data structure that starts with 100
million elements. We varied the batch size across experiments.  If the batch-insert
performance was slower than the non-batched insert done in a loop, the
non-batched insert number was reported. 

To measure the space usage, we vary the
number of elements and report the size.

Finally, we evaluate the serial batch-update algorithm from the Rewired PMA
(RMA)~\cite{de2019packed} with the provided test code and build scripts. For a
fair comparison, we ran the batch update algorithm for PMAs
from~\secref{batch-updates} on one core.

The RMA's provided tests use the numbers $[1, 2, \ldots, n]$ sampled without replacement where $n$ is the total number of elements
after the test. Although this is not exactly the same set as numbers as in our PMA experiments (with uniform random 40-bit numbers), the experiments are equivalent because both data structures are uncompressed, so only the ordering of the numbers matters.





\paragraph{Batch inserts on uniform random inputs}
~\figref{batch-micro} demonstrates that the throughput of parallel batch inserts
in the CPMA is on average \batchspeedup faster than in compressed
\pactrees. Similarly, parallel batch inserts in the PMA achieve on average
\batchspeedupoverpam faster throughput than in \ptrees.  The PMA's/CPMA's
cache-friendliness enables it to support faster updates than the theory
suggests.  As mentioned in~\secref{pma-prelim}, PMAs (and by extension, CPMAs)
support point updates in $O((\log^2(n))/B + \log(n))$ work. Trees theoretically
dominate PMAs for point updates: balanced binary trees support updates in
$O(1+ \log(n))$ work~\cite{CLRS}, and cache-friendly trees such as
B-trees~\cite{BayerMc72} support updates in $O(1+\log_B(n))$ work. However, in
practice, batch updates in a PMA/CPMA are faster than batch updates in trees
because the PMA/CPMA takes advantage of contiguous memory access.

\begin{table}[t]
  \begin{center}
    \begin{tabular}{@{}crrr@{}}
      \hline
      \textit{Batch size} & \textit{RMA}~\cite{de2019packed} & \textit{PMA} & \textit{PMA/RMA} \\
      \hline
      1---1E4             & 1.7E6                            & 2.2E6        & 1.3              \\
      1E5                 & 2.0E6                            & 2.4E6        & 1.2              \\
      1E6                 & 2.5E6                            & 3.2E6        & 1.3              \\
      1E7                 & 5.4E6                            & 6.5E6        & 1.2              \\
      \hline
    \end{tabular}
    \caption{Serial batch insert throughput (inserts per second) of the
      uncompressed PMA and RMA. We use point insertions for small
      batches when the batch update algorithm does not provide practical
      benefits.}
    \label{tab:rma-batch}
  \end{center}
\end{table} ~\tabref{rma-batch} evaluates the batch-insert
algorithm for uncompressed PMAs from~\secref{batch-updates} on one core compared
to the existing serial batch-insert algorithm for RMAs, an optimized version of
PMAs~\cite{de2019packed}. On average, the batch-insert algorithm in this paper
is about $1.2\times$ faster than the existing batch-insert algorithm for RMAs.

\paragraph{Batch inserts on skewed inputs}
Just as in the uniform random case, the CPMA outperforms C-PaC on small batches
and is slightly slower on large batches of skewed inserts.

The batch-parallel PMA is well-suited for the
case of all insertions targeting the same leaf. In contrast, for non batched PMAs, this is the worst case. The batch-insert PMA mitigates the worst case by (1)
sharing the work of searches between inserts, reducing overall work, and (2)
skipping levels of redistribution with larger batches, improving overall work
and parallelism. Due to these factors, the PMA/CPMA achieves higher throughput on
zipfian batch inserts compared to uniform random batch inserts as can be seen in \tabref{batch-ins-del}.

\paragraph{Batch deletes} On average, the PMA performs uniform random batch
deletions $1.9\times$ faster than uniform random batch insertions. Similarly,
the CPMA achieves $1.5\times$ higher throughput for uniform random batch
deletions compared to uniform random batch insertions, on average as can be seen in \tabref{batch-ins-del}. We see a
similar trend for the zipfian distribution. Batch deletions are faster than
batch insertions when the batch is large because deletes do not have to allocate
temporary space as they will never overflow the PMA leaves.

\begin{table*}[!ht]
  \begin{center}
    \setlength{\tabcolsep}{4pt}
 \resizebox{\textwidth}{!}{%
    \begin{tabular}{@{}crrrrrrrrrrrrrrrrrrrrr@{}}
      \hline
      & \multicolumn{10}{c}{\textit{Uniform}} &  &  \multicolumn{10}{c}{\textit{Zipfian}}
      \\
      \cmidrule{2-11}  \cmidrule{13-22}
      & \multicolumn{3}{c}{\textit{PMA}} & \phantom{a}      &
                                                              \multicolumn{3}{c}{\textit{CPMA}}  & \phantom{a} &  \multicolumn{2}{c}{\textit{CPMA/PMA}}
                                                            &  \phantom{a} & \multicolumn{3}{c}{\textit{PMA}} & \phantom{a}      &
                                                                                                                                   \multicolumn{3}{c}{\textit{CPMA}}
                                                                           & \phantom{a}      & \multicolumn{2}{c}{\textit{CPMA/PMA}}\\
      \cmidrule{2-4} \cmidrule{6-8}   \cmidrule{10-11} \cmidrule{13-15}
      \cmidrule{17-19} \cmidrule{21-22}
      \textit{Batch size} & \textit{Insert} & \textit{Delete} & \textit{D/I} & &
                                                                                 \textit{Insert}
                                                            &
                                                              \textit{Delete} &
                                                                                \textit{D/I}
                                                                                                               & &
                                                                                                                   \textit{Insert}
                                                                           &
                                                                             \textit{Delete} & &  \textit{Insert} & \textit{Delete} & \textit{D/I} & &
                                                                                                                                                       \textit{Insert}
                                                              &
                                                                \textit{Delete}
                                                                             &
                                                                               \textit{D/I}
                                                                               & \phantom{a} & \textit{Insert}
                                                                              &
                                                                                \textit{Delete}
      \\
      \hline
      1E1&1.8E6&1.8E6&1.0&&1.4E6&1.7E6&1.2&&0.8&0.9&&3.4E6&4.0E6&1.2&&2.7E6&3.6E6&1.3&&0.8&0.9\\
1E2&3.0E6&3.9E6&1.3&&2.6E6&3.2E6&1.2&&0.9&0.8&&3.6E6&4.2E6&1.2&&3.2E6&3.4E6&1.1&&0.9&0.8\\
1E3&9.0E6&1.3E7&1.5&&9.7E6&1.2E7&1.2&&1.1&0.9&&1.0E7&1.2E7&1.2&&1.1E7&1.2E7&1.1&&1.1&1.0\\
1E4&2.5E7&5.6E7&2.2&&3.3E7&5.1E7&1.5&&1.3&0.9&&2.7E7&3.5E7&1.3&&3.2E7&3.8E7&1.2&&1.2&1.1\\
1E5&4.1E7&8.6E7&2.1&&4.8E7&7.5E7&1.6&&1.2&0.9&&4.4E7&6.6E7&1.5&&7.2E7&8.7E7&1.2&&1.6&1.3\\
1E6&7.0E7&1.7E8&2.4&&1.1E8&1.7E8&1.6&&1.5&1.0&&7.8E7&1.4E8&1.8&&1.7E8&2.2E8&1.3&&2.2&1.6\\
1E7&1.0E8&4.0E8&3.9&&2.4E8&4.7E8&2.0&&2.3&1.2&&1.1E8&1.5E8&1.4&&3.1E8&4.6E8&1.5&&2.9&3.0\\
      \hline
    \end{tabular}
    }
    \caption{Parallel batch inserts and deletes (updates per second) for uniform and zipfian distribution for the PMA and CPMA.}
    \label{tab:batch-ins-del}
  \end{center}
\end{table*}

\paragraph{Range queries}
~\figref{map-range-micro} shows that the CPMA supports range queries between
\rangespeedupcpac faster than compressed \pactrees. Similarly, the PMA supports
range queries between \rangespeedupoverpam faster than \ptrees.  The PMA/CPMA is
faster to scan than compressed \pactrees because the PMA's/CPMA's contiguous
layout enables prefetching, while trees require pointer-chasing between tree
nodes. Furthermore, for small ranges, the PMA/CPMA are at least
\minsmallrangespeedup faster due to the pre-existing search layout optimizations
for PMAs, which are orthogonal to the optimizations in this
paper~\cite{WheatmanXu23}.

Furthermore, the CPMA supports range queries $1.3\times$ faster than
the PMA on the largest range because the CPMA's smaller size enables it to fetch
more elements before reaching memory bandwidth. However, the PMA is faster for
small range queries because of the added overhead of decompression in the CPMA.

\begin{table}[t]
  \begin{center}
    \resizebox{\columnwidth}{!}{%
    \begin{tabular}{@{}crrrrrrr@{}}
      \hline
      \begin{tabular}{@{}c@{}}\textit{Num.} \\ \textit{Elts.}  \end{tabular} & \textit{U-PaC} & \textit{PMA}   &
      \begin{tabular}{@{}c@{}}\textit{\underline{PMA}} \\ \textit{U-PaC}  \end{tabular}
                             &                \textit{C-PaC} & \textit{CPMA} & \begin{tabular}{@{}c@{}}\textit{\underline{CPMA}} \\ \textit{C-PaC}  \end{tabular} &
      \begin{tabular}{@{}c@{}}\textit{\underline{CPMA}} \\ \textit{PMA}  \end{tabular}
      \\
      \hline
      1E6                    & 8.07           & 11.82          & 1.46          &                      4.23 & 4.77 & 1.13 &  0.40 \\
      1E7                    & 8.12           & 10.51          & 1.30          &                      4.01 & 4.25 & 1.06  & 0.40 \\
      1E8                    & 8.09           & 11.36          & 1.40          &                      3.34 & 3.16 & 0.95  & 0.28 \\
      1E9                    & 8.07           & 9.89          & 1.23          &                      2.99 & 2.81 & 0.94  & 0.28 \\
      \hline
    \end{tabular}
    }
    \caption{Bytes per element in each of the data structures and compression ratios. The
      \texttt{sizeof(T)} is 8 bytes.}
    \label{tab:size-micro}
  \end{center}
\end{table}


\paragraph{Space usage}
~\tabref{size-micro} shows that CPMAs are similar in size to \cpac and are over
$2\times$ smaller than uncompressed PMAs.  The space savings of the compressed
data structures improves with the number of elements because the distance
between elements decreases as the number of elements increases. The CPMA uses
more space than \cpac for smaller inputs but less space than \cpac when the
input is sufficiently large (at least 100M elements) because the CPMA leaf size,
which defines the ratio of uncompressed to compressed elements, grows with the
number of elements.  As an uncompressed data structure, \ptrees take a fixed 32
bytes per element.



\subsection*{Evaluation on graph workloads}\label{sec:graph-eval}
We use the CPMA as the basis for a dynamic-graph container called \system and evaluate it on a suite of dynamic-graph workloads as an application benchmark for the CPMA.  
We first describe how \system processes dynamic graphs with a single
CPMA.  Then we present the results of the benchmark for \system, \cpac, and Aspen.




\paragraph{\system description} \system is built on a single batch-parallel CPMA
with delta compression and byte codes. It differs from traditional graph
representations because it uses only a single array to store both the vertex and
edge data.

To understand the distinction, consider the canonical Compressed Sparse
Row (CSR)~\cite{TinneyWa67} representation. For unweighted graphs, CSR uses two
arrays: an \defn{edge array} to store the edges in sorted order (by source and
then by destination), and a \defn{vertex array} to store offsets into the edge
array corresponding to the start of each vertex's neighbor list. The vertex
array saves space: the edge array then only needs to store destinations and not
sources.

In contrast, storing graphs in a CPMA takes only one array. Using a
CPMA, 
\system stores edges in 64-bit words by representing the source in the upper 32
bits and the destination in the lower 32 bits\footnote{ All of the tested graphs
  have fewer than $2^{32}$ vertices, so the edges fit in 64-bit words. If there
  are more than $2^{32}$ vertices, we can concatenate two 64-bit words to store
  each edge.}. The start of each vertex's neighbors is implicit and can be
restored with a search into the underlying CPMA. The delta compression in the
CPMA elides out the source vertex in all edges except for the edges in the
uncompressed PMA leaf heads and the first edge of each vertex.

\system supports batch updates and graph algorithms by adopting the popular
approach of phasing updates and algorithms separately~\cite{AmmarMcSa18,
  BusatoGrBo18, CaiLoSi12, EdgigerMcRi12, FengMeAm15, GreenBa16, MurrayMcIs16,
  shan2017accelerating, SenguptaSo17, SenguptaSuZh16, SuzuNiGa14, VoraGuXu17,
  WinterZaSt17}. It supports batch updates with one writer and therefore does
not use locks.

Finally, \system currently supports unweighted graphs because the CPMA is
currently a key store. \system also currently supports algorithms on undirected
graphs because it is built on a single CPMA, but it could be easily extended to
support algorithms on directed graphs with two CPMAs --- one for incoming edges
and one for outgoing edges\footnote{Since \system stores source/destination
  pairs, it can store directed graphs. However, many parallel graph algorithms
  require looping over both incoming and outgoing neighbor sets
  efficiently.}. Many graph algorithms (e.g., all the ones in this paper, among
others) can be run with only the graph topology. Future work includes extending
the CPMA to a key-value store which would allow \system to store weighted
graphs.

The CPMA under \system has a growing factor of $1.2\times$.

\paragraph{\othersystem and Aspen description} \othersystem and Aspen support dynamic-graph
processing with compressed trees (one per vertex) and enable concurrent
updates and graph algorithms without locking in functional mode. Since we are
not concurrently performing updates and algorithms, we use \cpac's and Aspen's in-place
unweighted modes for a fair comparison.

\paragraph{Systems setup}
All systems run the same algorithms via the Ligra interface, which is based on
the VertexSubset/EdgeMap abstraction~\cite{ShunBl13}.  Therefore, all algorithms
implemented with \othersystem and Aspen can be run on top of \system with minor
syntatic changes~\cite{ShunRoFo12, DhulBlSh18, DhulBlSh17}.

\paragraph{Datasets} 
\tabref{memory-footprints} lists the graphs used in the evaluation and their
sizes.  We tested on real social network graphs and a synthetic graph.  We used
a few social network graphs of various sizes: the \defn{LiveJournal}
(LJ)~\cite{LJ}, the \defn{Community Orkut} (CO)\cite{orkut}, the \defn{Twitter}
(TW)~\cite{BeamAsPa15}, and \defn{Friendster} (FS)~\cite{snapnets} graphs.
Additionally, we generated an \defn{Erd\H{o}s-R\'enyi} (ER) \\graph~\cite{erdos59a}
with $n = 10^7$ and $p = 5\cdot 10^{-6}$.

\subsubsection*{Graph algorithms}

\begin{figure}
        \centering
        \begin{tikzpicture}
                \begin{axis}[
                                ybar,
                                legend pos=north east,
                                ylabel={Speedup over \cpac},
                                symbolic x coords={PR, CC, BC, Average},
                                xtick=data,
                                legend columns=4,
                                width=8cm, height=4.5cm,
                                ymin=0,
                                ymax=1.7,
                                extra y ticks=1,
                                enlarge x limits=0.12,
                                extra y tick style={
                                                ymajorgrids=true,
                                                ytick style={/pgfplots/major tick length=0pt,},
                                                grid style={violet ,dashed,},
                                        },
                                every axis plot/.append style={fill},
                        ]

                        \foreach \a in {C-PaC, Aspen, F-Graph}{
                                        \addplot table [
                                                        discard if not={structure}{\a},
                                                        discard if not={graph}{Average},
                                                        x=kernel, y expr=1/\thisrow{normalized to CPaC}, col sep=tab
                                                ] {csvs/algorithm_times.tsv};

                                        \expandafter\addlegendentry\expandafter{\a}
                                }
                \end{axis}
        \end{tikzpicture}
         \vspace{-.5cm}
        \caption{Relative speedup of graph algorithms over \othersystem}
        \label{fig:intro-algs}
\end{figure}
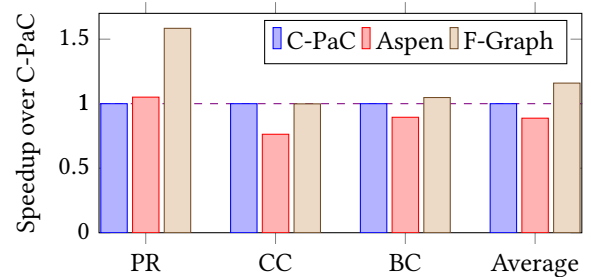

\begin{figure}
        \centering
        \begin{tikzpicture}
                \begin{axis}[
                                width=8cm, height=4.5cm,
                                axis lines = left,
                                legend columns=3,
                                xlabel = Batch size,
                                ylabel style={align=center},
                                ylabel = {Throughput \\(inserts / second)},
                                xmode=log,
                                ymode=log,
                                ymin=25000,
                                ymax=1000000000,
                                cycle list name=exotic,
                                legend pos=south east,
                                ytick={100000,1000000,10000000,100000000,1000000000},
                                xlabel shift={-6pt},
                        ]

                        \addplot[mark=star, safe-peach,mark options={scale=1.6}]
                        coordinates{(1E+1,3.7E+5)(1E+2,1.4E+6)(1E+3,6.2E+6)(1E+4,1.3E+7)(1E+5,2.2E+7)(1E+6,8.2E+7)(1E+7,2.1E+8)(1E+8,4.7E+8)};
                        \addlegendentry{\system}

                        \addplot[safe-pink, mark=square*]
                        coordinates{(10,146092.038)(100,1262626.263)(1000,3883495.146)(10000,4306632.214)(100000,10992634.93)(1000000,35688793.72)(10000000,70422535.21)(100000000,185322461.1)};
                        \addlegendentry{\othersystem}

                        \addplot[safe-teal, mark=diamond*]
                        coordinates{(10,1.10E+05)(100,8.01E+05)(1000,3.98E+06)(10000,6.16E+06)(100000,1.63E+07)(1000000,3.02E+07)(10000000,6.92E+07)(100000000,2.02E+08)};
                        \addlegendentry{Aspen}

                \end{axis}
        \end{tikzpicture}
         \vspace{-.5cm}
        \caption{Insert throughput as a function of batch size on the FS graph.}
        \label{fig:intro-insert}
\end{figure}
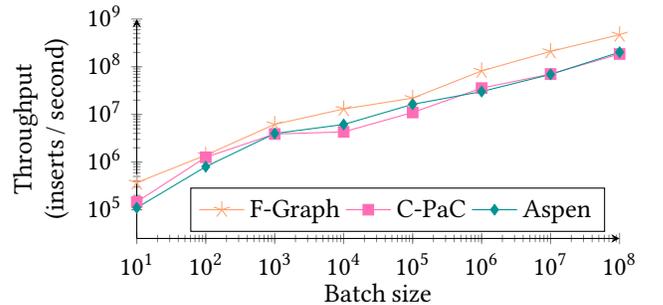


We evaluate the performance of \system, \cpac, and Aspen on three fundamental
graph algorithms: PageRank\cite{brin1998anatomy}\footnote{The PR implementation runs for a fixed
  number (10) of iterations.} (PR), connected components (CC), and single-source betweenness
centrality (BC).  ~\figref{intro-algs} presents the results of the evaluation,
and the full version of the paper contains all of the data. The algorithms are
from the Ligra~\cite{ShunBl13} distribution with minor cosmetic changes.  On
average, \system supports graph algorithms \graphalgspeedupcpac faster than
\cpac and \graphalgspeedupaspen faster than Aspen because \system stores the
graph contiguously in memory.

Traversals in graph kernels can be organized on a continuum depending on how
many long scans they contain, which depends on the order of vertices
accessed. On one extreme, \defn{arbitrary-order} algorithms such as PR access
vertices in any order and can be cast as a straightforward pass through the data
structure. On the other extreme, \defn{topology-order} algorithms such as BC
access vertices depending on the graph topology, and are therefore more likely
to incur cache misses by accessing a random vertex's neighbors. CC is in between
arbitrary order and topology order because it starts with large scans in the
beginning of the algorithm, but it converges to smaller scans as fewer vertices
remain under consideration.

Systems with a flat layout such as \system have an advantage when the algorithm
is closer to arbitrary order --- they support fast scans of neighbors because
all of the data is stored contiguously. For example, \system is $1.5\times$
faster than \cpac on average on PR. In contrast tree-based systems such as \cpac
incur more cache misses during large scans due to pointer chasing.

Since \system uses a single edge array in its flat layout, it must incur a fixed
cost to reconstruct the vertex array of offsets in all algorithms besides PR
(because PR accesses all of the edges in each iteration). The relative cost of
building the vertex array in \system compared to the cost of the algorithm
depends on the amount of other work in the algorithm.
For example, building the vertex array in \system takes about $10\%$ of the
total time in BC. The relative cost of building the offset array also depends on
the average degree: a higher average degree corresponds to a smaller overhead
compared to the cost of the algorithm. Finally, although this experiment
rebuilds the vertex array with each run of the algorithm, the vertex array could
be reused across computations (e.g., from different sources) if there have been
no updates.

\balance

\subsubsection*{Update throughput}
\system does not sacrifice updatability for its improved algorithm speed --- on
average, \system is $2\times$ faster than \othersystem and Aspen on batch
inserts.  ~\figref{intro-insert} shows that \system achieves faster updates than
\othersystem and Aspen despite the theoretical dominance of trees over PMAs in
terms of point and batch updates.

To evaluate insertion throughput, we first insert all edges from the FS graph
(the largest graph we tested on). We then add a new batch of directed edges
(with potential duplicates) to the existing graph in both systems.  To generate
edges for inserts, we sample directed edges from an \rmat
generator~\cite{ChakZhFa04} (with $a = 0.5$; $b = c = 0.1$; $d = 0.3$ to match the
distribution from the PaC-tree paper~\cite{DhulBlGu22}).

We note that the distribution of inserts is different here than in Section
\ref{sec:pma-eval}.  Here we see that even with a skewed distribution, while
traditionally challenging for PMAs, the batch parallel CPMA achieves good insert
throughput due to the work sharing in the batch insert algorithm.

\subsubsection*{Space usage}

Finally, we consider the space usage of \system, \othersystem, and
Aspen. ~\tabref{memory-footprints} shows that \system uses marginally less space
than \cpac and about \spacesavingsoveraspen the space that Aspen uses because
\system collocates small neighbor sets by using only one array to store all of
the data (rather than two levels of trees for vertices and edges in \othersystem
and Aspen).

\begin{table}[t]
  \centering
  \setlength{\tabcolsep}{4pt}
  \resizebox{\columnwidth}{!}{%
  \begin{tabular}{@{}lrrrrrrrr@{}}\toprule
    \textit{Graph} & \textit{N} & \textit{M} & \textit{\system} & \textit{\othersystem} & \textit{Aspen} &&  \textit{F/C} & \textit{F/A} \\
    \midrule
LJ&4.8&86&0.24&0.35&0.58&&0.69&0.41\\
CO&3.1&234&0.73&0.73&0.89&&1.00&0.82\\
ER&10&1000&3.74&3.80&5.17&&0.98&0.72\\
TW&62&2405&7.63&8.92&12.4&&0.86&0.62\\
FS&125&3612&13.21&14.99&22.76&&0.88&0.54\\
    \bottomrule
  \end{tabular}%
  }
  \caption{Graph sizes (N = number of vertices, M = number of edges, all in millions) and the memory used to store the graphs in all of the systems in Gigabytes. A number below 1 in the F/C or F/A column means that \system was smaller.}
  \label{tab:memory-footprints}
\end{table}


\section{Conclusion}\label{sec:conclusion}

This paper optimizes traditional PMAs with parallel batch updates and data
compression.  On average, the compressed PMA (CPMA) outperforms compressed trees (\pactrees)
by \batchspeedup on parallel batch
updates and \rangespeedupcpacavg on range queries due to the CPMA's cache-friendliness. The CPMA uses similar space compared to
compressed \pactrees and uses $2\times-3\times$ less space compared to
uncompressed representations. Compression enables the CPMA to scale better with
the number of cores compared to the PMA because its smaller size mitigates
memory bandwidth issues with reduced memory traffic.

To further demonstrate the real-world applicability of the CPMA, we introduce
\system, a dynamic-graph-processing system built on a single CPMA, and
compared it to \cpac, a state-of-the-art dynamic-graph-processing system built
on compressed \pactrees. We found that \system is \graphalgspeedupcpac faster on
graph algorithms, \graphbatchspeedupcpac faster on batch updates, and slightly
smaller when compared to \othersystem.

The empirical advantage of the CPMA over compressed \pactrees demonstrates the
importance of optimizing parallel data structures for the memory
subsystem. Specifically, the CPMA's array-based layout enables it to take
advantage of the speed of contiguous memory accesses. Despite the theoretical
prediction, the batch-parallel CPMA empirically overcomes the update/scan
tradeoff with compressed \pactrees due to its locality.

\clearpage
\bibliographystyle{ACM-Reference-Format}
\bibliography{sample}

\clearpage
\appendix
\section{Artifact Instructions}

This section summarizes how to download and use the code. The full details
(including how to compile the original binaries and reproduce the experiments in
the paper) can be found at the top level directory a pdf called "CPMA artifact readme" in both the git repo and the Zenodo (at
\url{https://zenodo.org/records/10222939}).

\paragraph{Machine specs}
Please use a machine with preinstalled \texttt{g++} (at least version 11) and
\texttt{git}. We have tested the artifact on an Amazon \texttt{c6i.metal} instance
running Ubuntu 20.04 with 128 threads and 256 GB of memory) and \texttt{g++} 11.4.

To run PAM/CPAM, you will also need jemalloc.

To make the plots, you will need python with matplotlib.

The test machine should have multiple threads but does not necessarily need 128
threads. In terms of memory, the known minimum necessary to run the graph
evaluation is 118 GB. This amount of memory is needed to run on the largest
graph we tested (Friendster).

The code should compile and run on non x86 machines, but the performance was
only tested on the machine above.

\paragraph{Get the code}
To get the code via git, clone the repo, go to the \texttt{for\_artifact}
branch, and set up the submodules:

\begin{lstlisting}[language=bash]
git clone https://github.com/wheatman/Packed-Memory-Array.git
cd Packed-Memory-Array
git checkout for_artifact
git submodule init
git submodule update
\end{lstlisting}

\begin{table*}
\centering
\small
\begin{tabular}{ ll }
\hline
 Description & Scripts \\ 
 \hline
 PMA/CPMA uniform batch inserts (all threads) & \texttt{run-fig-1.sh} \\  
 PMA/CPMA uniform range queries (all threads) & \texttt{run-fig-2.sh}   \\
  PMA/CPMA uniform batch inserts (serial) & \texttt{run-table-2.sh} (assuming you did the parallel ones via \texttt{run-fig-1.sh}) \\ 
  PMA/CPMA batch insert scalability (strong scaling) & \texttt{run-serial-fig-7.sh, sh run-parallel-fig-7.sh} \\
    PMA/CPMA range query scalability (strong scaling) & \texttt{run-serial-fig-8.sh, sh run-parallel-fig-8.sh} \\
  PMA/CPMA memory footprint on uniform dataset & \texttt{run-table-4.sh} \\
  CPMA graph evaluation & \texttt{run-graph-eval.sh} \\
 \hline
\end{tabular}
\caption{PMA/CPMA experiments in the paper and their associated scripts.}
\vspace{-.5cm}
\label{tab:artifact-map}
\end{table*}

\subsection*{PMA/CPMA API}

To use the PMA/CPMA for other purposes, follow the instructions from ``Get the
code'' above and add \texttt{\#include "CPMA.h"} to the top of your main test
driver.

The PMA/CPMA supports the following API:

\begin{itemize}
\item \texttt{uint64\_t size()}: Return the number of elements being stored in
  the PMA.
\item \texttt{CPMA()}: Construct an empty CPMA.
\item \texttt{CPMA(key\_type *start, key\_type *end)}: Construct a CPMA with the
  elements in the given range.
\item \texttt{bool has(key\_type e)}: Return true if the key \texttt{e} is in
  the PMA.
\item \texttt{bool insert(element\_type e)}: inserts the element \texttt{e} into
  the PMA, returns false if the key was already there.
\item \texttt{uint64\_t insert\_batch(element\_ptr\_type e, \\ uint64\_t
    batch\_size, bool sorted = false)}: Inserts a batch of elements of size
  \texttt{batch\_size}. Returns the number of elements added (not counting the
  ones that were already in the data structure).
\item \texttt{uint64\_t remove\_batch(key\_type *e, uint64\_t batch\_size, bool
    sorted = false)}: Removes a batch of elements of size \texttt{batch\_size}.
  Returns the number of elements removed (not counting the ones that were not in
  the data structure).
\item \texttt{bool remove(key\_type e)}: Removes the element with key \texttt{e}.
\item \texttt{uint64\_t get\_size()}: Returns the amount of memory (in bytes)
  used by the PMA.
\item \texttt{uint64\_t sum()}: Returns the sum of all elements in the PMA.
\item \texttt{key\_type max() / min()}: Returns the smallest or largest key
  stored in the PMA.
\item \texttt{bool map(F f)}: Runs function \texttt{f} on all elements in the
  PMA.
\item \texttt{parallel\_map(F f)}: Runs function \texttt{f} on all elements in
  the PMA in parallel.
\item \texttt{bool map\_range(F f, key\_type start\_key, key\_type end\_key)}:
  Runs function \texttt{f} on all elements with keys between \texttt{start\_key}
  and \texttt{end\_key}.
\item \texttt{uint64\_t map\_range\_length(F f, key\_type start, uint64\_t
    length)}: Runs function \texttt{f} on at most \texttt{length} elements
  starting from key at least \texttt{start}.
\item The PMA also supports iteration as it has begin and end functions, so you
  can perform operations like \texttt{for (auto el : pma)}. Note that this may
  be slower than using the map functions.
\end{itemize}

\subsection*{Relationship between scripts and data}

All of the scripts to run the PMA/CPMA are in the main \texttt{Packed-Memory-Array/} folder under \texttt{scripts}. \tabref{artifact-map} lists the type of PMA/CPMA experiment and the associated script to run it. The full documentation also includes instructions about how to run the other systems (PAM, U-PaC, C-PaC, Aspen).


\clearpage
\section{Data tables}\label{app:extra-data}

This section contains the data used to generate the plots
in~\secrefthree{intro}{cpma}{pma-eval}. The growing factor in the PMA/CPMA in the
microbenchmarks is $1.2\times$.  The \defn{growing factor} is the amount by
which the underlying array in the PMA grows when it becomes too dense. The
asymptotic bounds of the PMA still hold as long as the growing factor is a
constant greater than 1.  We chose the growing factor based on the
microbenchmarks in~\secref{growing-appendix}.

\subsection*{\secref{intro}}
\FloatBarrier

\tabref{batch-micro-table} contains the data used to
generate~\figref{batch-micro}, and~\tabref{range-micro-table} contains the data
used to generate~\figref{map-range-micro}.~\tabref{intro-cache-misses} reports
the cache misses of each data structure mentioned in~\secref{intro}.

\begin{table*}[!ht]
  \begin{center}
    \begin{tabular}{@{}crrrrrrrrr@{}}
      \hline
     \begin{tabular}{@{}c@{}}\textit{Batch} \\ \textit{size}  \end{tabular}& \textit{\ptrees} & \textit{U-PaC} &
                                             \textit{PMA}
        & \begin{tabular}{@{}c@{}}\textit{\underline{PMA}} \\ \textit{\ptrees}  \end{tabular}
      & \begin{tabular}{@{}c@{}}\textit{\underline{PMA}} \\ \textit{U-PaC}  \end{tabular}
      &              \textit{\cpac}
      & \textit{CPMA} & \begin{tabular}{@{}c@{}}\textit{\underline{CPMA}} \\ \textit{\cpac}  \end{tabular}   & \begin{tabular}{@{}c@{}}\textit{\underline{CPMA}} \\ \textit{PMA} \end{tabular}                                \\
      \hline
1E1&3.4E5&1.9E5&1.8E6&5.2&9.3&3.0E5&1.4E6&4.7&0.8\\
1E2&2.2E6&2.4E6&3.0E6&1.4&1.3&1.8E6&2.6E6&1.5&0.9\\
1E3&9.7E6&4.8E6&9.0E6&0.9&1.9&3.4E6&9.7E6&2.8&1.1\\
1E4&1.7E7&5.9E6&2.5E7&1.5&4.3&4.2E6&3.3E7&7.9&1.3\\
1E5&3.4E7&1.1E7&4.1E7&1.2&3.7&7.2E6&4.8E7&6.7&1.2\\
1E6&5.0E7&6.1E7&7.0E7&1.4&1.2&4.1E7&1.1E8&2.6&1.5\\
1E7&8.5E7&3.5E8&1.0E8&1.2&0.3&2.7E8&2.4E8&0.9&2.3\\
      \hline
    \end{tabular}
    \caption{Parallel batch insertion throughput (inserts per second) on all
      cores in \ptrees, \pactrees, and the PMA/CPMA.}
    \label{tab:batch-micro-table}
  \end{center}
\end{table*}

\begin{table*}[!ht]
  \begin{center}
    \begin{tabular}{@{}crrrrrrrrr@{}}
      \hline
       \begin{tabular}{@{}c@{}}\textit{Avg.} \\ \textit{len.}  \end{tabular} &\textit{\ptrees}&\textit{U-PaC} &
                                                            \textit{PMA}
      &\begin{tabular}{@{}c@{}}\textit{\underline{PMA}} \\ \textit{\ptrees}  \end{tabular}& \begin{tabular}{@{}c@{}}\textit{\underline{PMA}} \\ \textit{U-PaC}  \end{tabular}       & \textit{\cpac} & \textit{CPMA}
      & \begin{tabular}{@{}c@{}}\textit{\underline{CPMA}} \\ \textit{\cpac}  \end{tabular}   & \begin{tabular}{@{}c@{}}\textit{\underline{CPMA}} \\ \textit{PMA} \end{tabular}                                   \\
      \hline
      6E0      &        1.9E8         & 2.0E8                & 1.7E9          &    8.9    & 8.3                  & 1.8E8  & 8.1E8  & 4.4  &  0.5 \\
      5E1      &        4.9E8         & 9.5E8                & 6.6E9          &    13.6   & 7.0                   & 8.5E8  & 5.1E9  & 6.0  &   0.8 \\
      4E2      &        6.0E8         & 2.1E9                & 1.3E10         &    21.8   & 6.1                    & 1.5E9  & 1.5E10 & 10.3 &   1.2 \\
      3E3      &        6.2E8         & 1.0E10               & 1.6E10         &    24.9   & 1.5                   & 7.0E9  & 2.2E10 & 3.1  &   1.4 \\
      2E4      &        6.5E8         & 1.6E10               & 1.7E10         &    26.7   & 1.1                    & 1.6E10 & 2.4E10 & 1.5  &   1.4 \\
      2E5      &        6.8E8         & 1.8E10               & 1.8E10         &    27.1   & 1.0                    & 1.9E10 & 2.4E10 & 1.3  &  1.3 \\
      2E6      &        6.9E8         & 1.9E10               & 1.9E10         &    27.4   & 1.0                   & 1.9E10 & 2.4E10 & 1.2  &  1.3 \\
      \hline
    \end{tabular}
    \caption{Range query throughput (elements per second) on all
      cores in \ptrees, \pactrees, and the PMA/CPMA.}
    \label{tab:range-micro-table}
  \end{center}
\end{table*}

\subsection*{\secref{cpma}}

\tabref{batch-insert-scaling-table} contains the data
for~\figref{strong_scaling}, and~\tabref{range-query-scaling-table} contains the
data for~\figref{map_range_scaling}.

\begin{table}[t]
  \centering
  \begin{tabular}{@{}lrrrrrr@{}}\toprule
    \textit{Cores} & \begin{tabular}{@{}c@{}}\textit{PMA} \\ \textit{throughput}  \end{tabular} & \begin{tabular}{@{}c@{}}\textit{PMA} \\ \textit{speedup}  \end{tabular}&  &\begin{tabular}{@{}c@{}}\textit{CPMA} \\ \textit{throughput}  \end{tabular} & \begin{tabular}{@{}c@{}}\textit{CPMA} \\ \textit{speedup}  \end{tabular} \\
    \midrule
1&3.0E6&1.0&&2.6E6&1.0\\
2&4.9E6&1.6&&4.6E6&1.8\\
4&9.2E6&3.1&&9.1E6&3.5\\
8&1.8E7&5.9&&1.8E7&6.8\\
16&2.7E7&9.1&&3.4E7&13.2\\
32&4.1E7&13.7&&5.6E7&21.7\\
64&5.3E7&17.9&&8.5E7&33.1\\
64h&4.8E7&16.3&&1.1E8&43.3\\
    \bottomrule
  \end{tabular}%
  \caption{Batch insert scalability as a function of the number of cores. We use
    64 to denote all physical cores and 64h to denote all 128 hyperthreads.}
  \label{tab:batch-insert-scaling-table}
\end{table}

\begin{table}
  \centering
  \begin{tabular}{@{}lrrrrrr@{}}\toprule
    \textit{Cores} & \begin{tabular}{@{}c@{}}\textit{PMA} \\ \textit{throughput}  \end{tabular} & \begin{tabular}{@{}c@{}}\textit{PMA} \\ \textit{speedup}  \end{tabular}&  &\begin{tabular}{@{}c@{}}\textit{CPMA} \\ \textit{throughput}  \end{tabular} & \begin{tabular}{@{}c@{}}\textit{CPMA} \\ \textit{speedup}  \end{tabular} \\
    \midrule
    1&4.5E8&1.0&&2.0E8&1.0\\
    2&8.6E8&1.9&&4.2E8&2.1\\
    4&1.7E9&3.8&&8.5E8&4.2\\
    8&3.4E9&7.5&&1.7E9&8.4\\
    16&6.7E9&14.7&&3.4E9&16.8\\
    32&1.3E10&27.6&&6.8E9&33.6\\
    64&1.7E10&37.4&&1.4E10&66.9\\
    64h&1.9E10&41.5&&2.4E10&117.8\\
    \bottomrule
  \end{tabular}%
  \caption{Range query scalability as a function of the number of cores. We use
    64 to denote all physical cores and 64h to denote all 128 hyperthreads.}
  \label{tab:range-query-scaling-table}
\end{table}


\subsection*{\secref{graph-eval}}

\paragraph{Batch inserts with zipfian distribution}

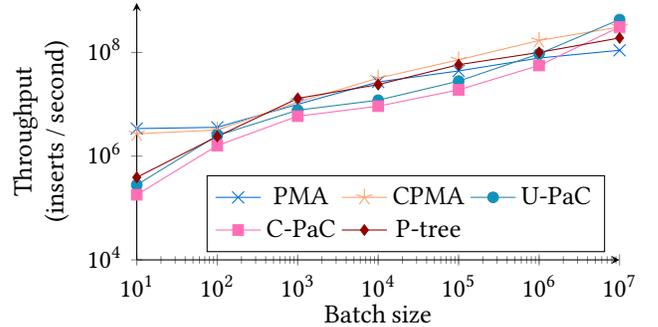
\begin{figure}
  \centering
  \begin{tikzpicture}
    \begin{axis}[
        width=8cm, height=5cm,
        axis lines = left,
        xlabel = Batch size,
        ylabel style={align=center},
        ylabel = {Throughput\\ (inserts / second)},
        xmode=log,
        ymode=log,
        ymin=10000,
        ymax=900000000,
        cycle list name=exotic,
        legend pos=south east,
        legend columns=3,
        xlabel shift={-6pt},
      ]

      \addplot[mark=x, safe-cerulean,mark options={scale=1.6}]
      coordinates{(10,3.4E+6)(100,3.6E+6)(1000,1.0E+7)(10000,2.7E+7)(100000,4.4E+7)(1000000,7.8E+7)(10000000,1.1E+8)};
      \addlegendentry{PMA}

      \addplot[mark=star, safe-peach,mark options={scale=1.6}]
      coordinates{(10,2.7E+6)(100,3.2E+6)(1000,1.1E+7)(10000,3.2E+7)(100000,7.2E+7)(1000000,1.7E+8)(10000000,3.1E+8)};
      \addlegendentry{CPMA}

      \addplot
      coordinates{(10,2.8E+5)(100,2.5E+6)(1000,7.7E+6)(10000,1.2E+7)(100000,2.8E+7)(1000000,9.4E+7)(10000000,4.3E+8)};
      \addlegendentry{U-PaC}

      \addplot[safe-pink, mark=square*]
      coordinates{(10,1.8E+5)(100,1.6E+6)(1000,5.9E+6)(10000,9.2E+6)(100000,1.9E+7)(1000000,5.6E+7)(10000000,3.1E+8)};
      \addlegendentry{C-PaC}

     \addplot[safe-brick, mark=diamond*]
      coordinates{(10,3.9E+5)(100,2.4E+6)(1000,1.3E+7)(10000,2.4E+7)(100000,5.8E+7)(1000000,1.0E+8)(10000000,1.9E+8)};
      \addlegendentry{\ptree}

    \end{axis}
  \end{tikzpicture}
  \vspace{-.8cm}
  \caption{Insert throughput as a function of batch size with batches generated
    from a zipfian distribution.}
  \label{fig:batch-micro-zipf}
\end{figure}

\begin{table*}
  \begin{center}
    \begin{tabular}{@{}crrrrrrrrr@{}}
      \hline
     \begin{tabular}{@{}c@{}}\textit{Batch} \\ \textit{size}  \end{tabular}& \textit{\ptrees} & \textit{U-PaC} &
                                             \textit{PMA}
        & \begin{tabular}{@{}c@{}}\textit{\underline{PMA}} \\ \textit{\ptrees}  \end{tabular}
      & \begin{tabular}{@{}c@{}}\textit{\underline{PMA}} \\ \textit{U-PaC}  \end{tabular}
      &              \textit{\cpac}
      & \textit{CPMA} & \begin{tabular}{@{}c@{}}\textit{\underline{CPMA}} \\ \textit{\cpac}  \end{tabular}   & \begin{tabular}{@{}c@{}}\textit{\underline{CPMA}} \\ \textit{PMA} \end{tabular}                                \\
      \hline
1.0E1&3.9E5&2.8E5&3.4E6&8.7&12.2&1.8E5&2.7E6&15.5&0.8\\
1.0E2&2.4E6&2.5E6&3.6E6&1.5&1.5&1.6E6&3.2E6&2.0&0.9\\
1.0E3&1.3E7&7.7E6&1.0E7&0.8&1.4&5.9E6&1.1E7&1.9&1.1\\
1.0E4&2.4E7&1.2E7&2.7E7&1.1&2.3&9.2E6&3.2E7&3.5&1.2\\
1.0E5&5.8E7&2.8E7&4.4E7&0.8&1.6&1.9E7&7.2E7&3.8&1.6\\
1.0E6&1.0E8&9.4E7&7.8E7&0.7&0.8&5.6E7&1.7E8&3.0&2.2\\
1.0E7&1.9E8&4.3E8&1.1E8&0.5&0.2&3.1E8&3.1E8&1.0&2.9\\
      \hline
    \end{tabular}
    \caption{Parallel batch insertion throughput (inserts per second) for
      inserts from a zipfian distribution on all
      cores in \ptrees, \pactrees, and the PMA/CPMA.}
    \label{tab:batch-micro-zipf}
  \end{center}
\end{table*}


~\figref{batch-micro-zipf} and~\tabref{batch-micro-zipf} contain the data for
zipfian batch inserts.



\paragraph{Evaluation on graph workloads}

\tabref{all-systems-kernels} contains the data for~\figref{intro-algs}, and
~\tabref{batch-graph-table} contains the data for~\figref{intro-insert}.

\begin{table*}
  \begin{center}
    \setlength{\tabcolsep}{4pt}
  \resizebox{\textwidth}{!}{%
    \begin{tabular}{@{}crrrrrrrrrrrrrrrrr@{}}
      \hline
                     & \multicolumn{5}{c}{\textit{PR}} & \phantom{a}      & \multicolumn{5}{c}{\textit{CC}} & \phantom{a} & \multicolumn{5}{c}{\textit{BC}}                                                                                         \\
      \cmidrule{2-6} \cmidrule{8-12} \cmidrule{14-18}
      \textit{Graph} & \textit{Aspen}   &  \textit{\othersystem}       & \textit{\system} & \begin{tabular}{@{}c@{}}\textit{\underline{Aspen}} \\ \textit{\system}  \end{tabular} & \begin{tabular}{@{}c@{}}\textit{\underline{\othersystem}} \\ \textit{\system}  \end{tabular}                   &             & \textit{Aspen}   &  \textit{\othersystem}       & \textit{\system} & \begin{tabular}{@{}c@{}}\textit{\underline{Aspen}} \\ \textit{\system}  \end{tabular} & \begin{tabular}{@{}c@{}}\textit{\underline{\othersystem}} \\ \textit{\system}  \end{tabular}  & & \textit{Aspen}   &  \textit{\othersystem}       & \textit{\system} & \begin{tabular}{@{}c@{}}\textit{\underline{Aspen}} \\ \textit{\system}  \end{tabular} & \begin{tabular}{@{}c@{}}\textit{\underline{\othersystem}} \\ \textit{\system}  \end{tabular} \\
      \hline
LJ&0.22&0.18&0.13&1.69&1.37&&0.07&0.05&0.08&0.91&0.57&&0.08&0.06&0.07&1.26&0.93\\
CO&0.35&0.39&0.27&1.31&1.46&&0.09&0.08&0.07&1.26&1.00&&0.09&0.07&0.07&1.26&0.98\\
ER&2.97&3.42&1.95&1.52&1.75&&1.02&0.54&0.40&2.52&1.34&&0.21&0.21&0.19&1.09&1.09\\
TW&9.52&10.77&6.98&1.36&1.54&&2.00&2.08&1.77&1.13&1.17&&1.07&1.12&1.18&0.91&0.95\\
FS&23.30&26.47&13.94&1.67&1.90&&4.97&5.66&3.81&1.31&1.49&&3.02&3.05&2.15&1.40&1.42\\
      \hline
    \end{tabular}
    }
    \caption{Running times (seconds) of Aspen, \othersystem, and  \system on PR, CC, and BC
      with all (\numcores) threads. A number above 1 in the ratio
      columns means that \system was faster.}
      \vspace{-.7cm}
    \label{tab:all-systems-kernels}
  \end{center}
\end{table*}

\begin{table}
  \begin{center}
    \begin{tabular}{@{}crrrrrr@{}}
      \hline
     \begin{tabular}{@{}c@{}}\textit{Batch} \\ \textit{size}  \end{tabular}& \textit{Aspen} & \textit{\othersystem} &
                                             \textit{\system}
        && \begin{tabular}{@{}c@{}}\textit{\underline{\system}} \\ \textit{Aspen}  \end{tabular}
      & \begin{tabular}{@{}c@{}}\textit{\underline{\system}} \\ \textit{\othersystem}  \end{tabular}
                                      \\
      \hline
1E1&1.10E05&1.5E5&3.7E5&&3.4&2.5\\
1E2&8.01E05&1.3E6&1.4E6&&1.7&1.1\\
1E3&3.98E06&3.9E6&6.2E6&&1.6&1.6\\
1E4&6.16E06&4.3E6&1.3E7&&2.2&3.1\\
1E5&1.63E07&1.1E7&2.2E7&&1.3&2.0\\
1E6&3.02E07&3.6E7&8.2E7&&2.7&2.3\\
1E7&6.92E07&7.0E7&2.1E8&&3.0&2.9\\
1E8&2.02E08&1.9E8&4.7E8&&2.3&2.5\\
      \hline
    \end{tabular}
    \caption{Parallel batch insertion throughput (inserts per second) on all
      cores in Aspen, \othersystem, and \system.  The base graph is the FS
      graph.  The new insertions are sampled from the \rmat distribution.}
    \label{tab:batch-graph-table}
  \end{center}
\end{table}


\clearpage

\section{Growing factor sensitivity}\label{sec:growing-appendix}

We evaluate how the space usage, batch insertion throughput, and scan
throughput of the CPMA changes with the growing factor.

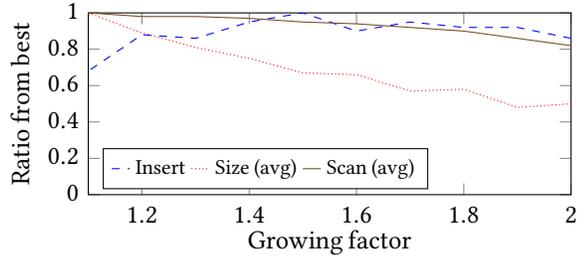
\begin{figure}[t]
\begin{tikzpicture}
  \begin{axis}[xlabel=Growing factor,ylabel=Ratio from best,
    width=8cm, height=4cm,
    legend columns=3,
    ymin=0,
    ymax=1,
    xmin=1.1,
    xmax=2,
    legend pos=south west,
    no marks,
        legend style={font=\footnotesize},
    legend image post style={scale=0.4},
    xlabel shift={-6pt},]
\addplot [dashed, blue] table [
    x=growing factor, y=insert time, col sep=tab
    ] {csvs/growing_factor_analysis.tsv};
    \addlegendentry{Insert}

\addplot [densely dotted, red] table [
    x=growing factor, y=average size, col sep=tab
    ] {csvs/growing_factor_analysis.tsv};
    \addlegendentry{Size (avg)}

\addplot table [
    x=growing factor, y=average sum time, col sep=tab
    ] {csvs/growing_factor_analysis.tsv};
    \addlegendentry{Scan (avg)}
\end{axis}
\end{tikzpicture}
\vspace{-.4cm}
\caption{Effect of growing factor on performance and size.}
\label{fig:growing_factor_overall}
\end{figure}

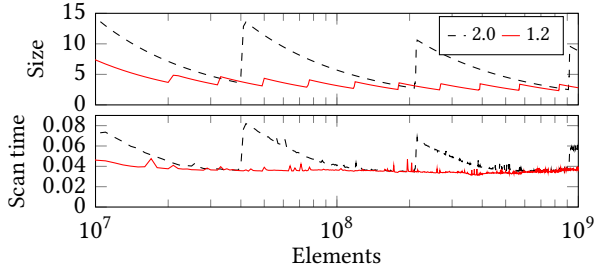
\begin{figure}[t]
\begin{tikzpicture}

    \begin{groupplot}[
        group style={
            group size=1 by 2,
            x descriptions at=edge bottom,
            vertical sep=4pt,
        },
    xmin=10000000,
    xmax=1000000000,
    xmode=log,
        legend pos=north east,
    no marks,
    label style={font=\small},
    tick label style={font=\scriptsize},
    scaled y ticks=false,
        max space between ticks=1000pt,
        try min ticks=4,
        width=8cm, height=2.8cm,
        xlabel shift={-6pt},
    ]

    \nextgroupplot[
  ylabel=Size,
  ylabel style={align=center},
    legend columns=2,
    ymin=0,
    ymax=15,
    legend pos=north east,
    no marks,
    legend style={font=\footnotesize},
    legend image post style={scale=0.4},
    ]
\addplot [dashed] table [
    x expr=(\thisrow{batch number}+1)*1000000, y=2.0 size, col sep=tab
    ] {csvs/growing_factor_analysis_per_batch.tsv};
    \addlegendentry{2.0}
\addplot table [
    x expr=(\thisrow{batch number}+1)*1000000, y=1.2 size, col sep=tab
    ] {csvs/growing_factor_analysis_per_batch.tsv};
    \addlegendentry{1.2}

    \nextgroupplot[
  xlabel=Elements,ylabel=Scan time,
    ylabel style={align=center},
    ymin=0,
    ymax=.09,
    legend image post style={scale=0.4},
    yticklabel style={
        /pgf/number format/fixed,
        /pgf/number format/precision=2
},
    ]
\addplot [dashed] table [
    x expr=(\thisrow{batch number}+1)*1000000, y expr=\thisrow{2.0 sum time}*1000, col sep=tab
    ] {csvs/growing_factor_analysis_per_batch.tsv};
\addplot table [
    x expr=(\thisrow{batch number}+1)*1000000, y expr=\thisrow{1.2 sum time}*1000, col sep=tab
    ] {csvs/growing_factor_analysis_per_batch.tsv};
    \end{groupplot}
\end{tikzpicture}
\vspace{-.4cm}
\caption{Size (bytes) and scan time (ns) per element after each batch insertion in CPMAs with different growing factors.}
\label{fig:growing_factor_growth}
\end{figure}

To measure the effect of growing factor on the CPMA, we performed the following
experiment on CPMAs with growing factors
$1.1\times, 1.2\times, \ldots, 2.0\times$. We started with an empty CPMA and
added 1 billion elements in parallel batches of 1 million elements each (for a
total of 1,000 batches). After each batch, we measured the space usage of the
CPMA and performed a parallel scan over all of the elements. We also measure the
batch insertion throughput as a function of growing factor.

~\figref{growing_factor_overall} demonstrates that a smaller growing factor
results in smaller average space usage and therefore better average scan
performance because smaller sizes require less memory
traffic.~\figref{growing_factor_growth} shows that the growing factor bounds the
worst-case space usage of the CPMA: a CPMA with a higher growing factor has a
higher worst-case space usage. However, the exact space usage and scan
performance depend not only on the growing factor but also on the state of the
CPMA (i.e., how far it is from a growth).

Moreover, the relationship between insert time and growing factor is not as
straightforward as the relationship between size/scan and growing
factor. \figref{growing_factor_overall} shows that the CPMA with growing factor
$1.5\times$ achieves the best insertion throughput.  Small growing factors
(e.g., $1.1\times$) increase the number of array copies since the CPMA incurs
more growths, but also decrease the size of the array which improves other parts
of inserts such as the binary search and rebalances.  On the other hand, large
growing factors (e.g., $2\times$) have fewer array copies but longer searches,
which contribute to more expensive inserts.


\end{document}